\documentclass[final]{l4dc2024}

\usepackage{times}
\usepackage{hyperref}       
\usepackage{booktabs}       
\usepackage{amsfonts}       
\usepackage{nicefrac}       
\usepackage{microtype}      
\usepackage{xcolor}         
\usepackage{mathtools}
\usepackage{cleveref}
\usepackage{algorithmic}
\usepackage{algorithm}
\usepackage{wrapfig}

\newcommand{\algo}{{LYGE}}

\newcommand{\eg}{\textit{e.g.}}
\newcommand{\ie}{\textit{i.e.}}

\newcommand{\xg}{x_\mathrm{goal}}

\usepackage{scalerel,stackengine}
\stackMath
\newcommand\reallywidehat[1]{%
\savestack{\tmpbox}{\stretchto{%
  \scaleto{%
    \scalerel*[\widthof{\ensuremath{#1}}]{\kern-.6pt\bigwedge\kern-.6pt}%
    {\rule[-\textheight/2]{1ex}{\textheight}}%
  }{\textheight}%
}{0.5ex}}%
\stackon[1pt]{#1}{\tmpbox}%
}

\title[Lyapunov-guided Exploration]{Learning to Stabilize High-dimensional Unknown Systems Using Lyapunov-guided Exploration}
\usepackage{times}

\coltauthor{\Name{Songyuan Zhang} \Email{szhang21@mit.edu}\and
 \Name{Chuchu Fan} \Email{chuchu@mit.edu}\\
 \addr Department of Aeronautics and Astronautics, Massachusetts Institute of Technology, Cambridge, MA, USA}

\begin{document}

\maketitle

\begin{abstract}%
    Designing stabilizing controllers is a fundamental challenge in autonomous systems, particularly for high-dimensional, nonlinear systems that can hardly be accurately modeled with differential equations. The Lyapunov theory offers a solution for stabilizing control systems, still, current methods relying on Lyapunov functions require access to complete dynamics or samples of system executions throughout the entire state space. Consequently, they are impractical for high-dimensional systems. This paper introduces a novel framework, \textbf{LY}apunov-\textbf{G}uided \textbf{E}xploration (\algo), for learning stabilizing controllers tailored to high-dimensional, unknown systems. \algo\ employs Lyapunov theory to iteratively guide the search for samples during exploration while simultaneously learning the local system dynamics, control policy, and Lyapunov functions. We demonstrate its scalability on highly complex systems, including a high-fidelity F-16 jet model featuring a 16D state space and a 4D input space. Experiments indicate that, compared to prior works in reinforcement learning, imitation learning, and neural certificates, \algo\ reduces the distance to the goal by $50\%$ while requiring only $5\%$ to $32\%$ of the samples. Furthermore, we demonstrate that our algorithm can be extended to learn controllers guided by other certificate functions for unknown systems.\footnote{Project website: \href{https://mit-realm.github.io/lyge-website/}{https://mit-realm.github.io/lyge-website/}. The appendix can be found on the project website.}
\end{abstract}


\section{Introduction}

Designing stabilizing controllers for high-dimensional systems with potentially unknown dynamics is essential in autonomous systems, where Lyapunov-based control design plays a significant role~\citep{parrilo2000structured}. In recent years, methods have been proposed to automatically construct Lyapunov functions and control Lyapunov functions (CLFs) for systems of varying complexity, encompassing both optimization and learning-based methods~\citep{giesl2015review,dawson2022survey}. However, existing methods encounter two major challenges: \emph{scalability}, \ie, applicability to high-dimensional systems, and \emph{model transparency}, \ie, the necessity of knowing the system dynamics.

Traditional optimization-based control design involves finding controllers and Lyapunov functions by solving a sequence of semi-definite programming (SDP) problems~\citep{parrilo2000structured,majumdar2013control,ahmadi2016some}. However, scalability to high-dimensional systems is hindered by the exponential growth of the number of decision variables with system dimension~\citep{lofberg2009pre} and numerical problems~\citep{permenter2018partial}, including strict feasibility and numerical reliability issues. Although recent advancements have produced scalable SDP solvers~\citep{yurtsever2021scalable}, they depend on assumptions such as the sparsity of the decision matrix. Neural network (NN)-based representations of Lyapunov functions have gained popularity~\citep{chang2019neural,han2020actor,chang2021stabilizing,dawson2022safe} and can to some extent alleviate dimensionality limitations when finding a CLF. However, due to their reliance on state-space sampling, learning techniques still face exponential growth in sample complexity for high-dimensional systems.

Additionally, most optimization-based and learning-based methods require system dynamics to be known as ordinary differential equations (ODEs), limiting their applicability and practicality for real-world systems. For instance, the F-16 fighter jet model~\citep{heidlauf2018verification} investigated in this paper is represented as a combination of look-up tables, block diagrams, and C programs. Accurately describing the complex behavior of the system using an ODE is highly challenging. For systems with unknown dynamics, previous works have attempted to first use system identification, (\eg, learn the ODE model with NNs), then find a controller with a CLF~\citep{dai2021lyapunov,zhou2022neural}. However, using a single NN to fit the dynamics of the entire state space of a high-dimensional system requires a vast number of training samples to cover the entire state space, and these surrogate NN models can exhibit substantial prediction errors in sparsely sampled regions of the state-action space.

Our work is motivated by the fact that for high-dimensional systems, collecting data across the entire state space is both infeasible and unnecessary, as only a small subset of the state space is reachable for the agent starting from a set of initial conditions. Therefore, obtaining a model for the entire state space is excessive \citep{kamalapurkar2016efficient}. Instead, we learn a model valid only in the reachable subset of the state space and update the controller according to guidance, e.g., the learned CLF, to continuously expand the subset towards the goal. Constructing such a subset is non-trivial, so we assume access to some imperfect and potentially unstable demonstrations as initial guidance for reachable states. Starting from these demonstrations, our objective is to update the controller, guide exploration of necessary regions in the state space, and ultimately stabilize the system at the goal.

To achieve this, we propose a novel framework, \textbf{LY}apunov-\textbf{G}uided \textbf{E}xploration (\algo), to jointly learn the system dynamics in the reachable subset, a controller, and a CLF to guide the exploration of the controller in high-dimensional unknown systems. We iteratively learn the dynamics in the reachable states using past experience, update the CLF and the controller, and perform exploration to expand the subset toward the goal. Upon convergence, we obtain a stabilizing controller for the high-dimensional unknown system. The main \textbf{contributions} of the paper are: 
1) We propose a novel framework, \algo, to learn stabilizing controllers for high-dimensional unknown systems. Guided by a learned CLF, \algo\ explores only the \emph{useful} subset, thus addressing the scalability and model transparency problems; 
2) We show that the proposed algorithm learns a stabilizing controller; 
3) We conduct experiments on benchmarks including Inverted Pendulum, Cart Pole, Cart II Pole~\citep{gym}, Neural Lander~\citep{shi2019neural}, and the F-$16$ model~\citep{heidlauf2018verification} with two tasks. Our results show that our learned controller outperforms other reinforcement learning (RL), imitation learning (IL), and neural-certificate-based algorithms in terms of stabilizing the systems while reducing the number of samples by $68\%$ to $95\%$.

\section{Related Work}

\paragraph{Control Lyapunov Functions.}
Our work builds on the widely used Lyapunov theory for designing stabilizing controllers. Classical CLF-based controllers primarily rely on hand-crafted CLFs~\citep{choi2020reinforcement,castaneda2021gaussian} or Sum-of-Squares (SoS)-based SDPs~\citep{parrilo2000structured,majumdar2013control,ahmadi2016some,long2023distributionally}. However, these approaches require known dynamics and struggle to generalize to high-dimensional systems due to the exponential growth of decision variables and numerical issues~\citep{lofberg2009pre,permenter2018partial}. To alleviate these limitations, recent work uses NN to learn Lyapunov functions~\citep{richards2018lyapunov,abate2020formal,abate2021fossil,gaby2021lyapunov} and stabilizing controllers~\citep{chang2019neural,mehrjou2021neural,dawson2022safe,farsi2022piecewise,zhang2023compositional,wang2023physics,min2023data}. Most of these works sample states in the entire state space and apply supervised learning to enforce the CLF conditions. They either assume knowledge of the dynamics or attempt to fit the entire state space's dynamics~\citep{dai2021lyapunov,zhou2022neural}, making it difficult to generalize to high-dimensional real-world scenarios where the number of required samples grows exponentially with the dimensions. In contrast, our algorithm can handle unknown dynamics and does not suffer from the curse of dimensionality caused by randomly sampling states in the entire state space.

\paragraph{Reinforcement Learning (RL) and Optimal Control.}
RL and optimal control have demonstrated strong capabilities on problems without knowledge of the dynamics, particularly in hybrid systems~\citep{schulman2015trust,schulman2017proximal,rosolia2019learning,so2023solving}. However, they struggle to provide results of the closed-loop system's stability. Additionally, hand-crafted reward functions and sample inefficiency impede the generalization of RL algorithms to complex environments. Recent works in the learning for control domain aim to solve this problem by incorporating certificate functions into the RL process~\citep{berkenkamp2017safe,chow2018lyapunov,cheng2019end,han2020actor,chang2021stabilizing,zhao2021model,qin2021density}. Nevertheless, they suffer from limitations such as handcrafted certificates~\citep{berkenkamp2017safe} and balancing CLF-related losses with RL losses \citep{han2020actor,chang2021stabilizing}. Unlike these approaches, our algorithm learns the CLF from scratch without prior knowledge of the dynamics or CLF candidates and provides a structured way to design loss functions rather than relying on reward functions. Furthermore, we can demonstrate the stability of the closed-loop system using the learned CLF.

\paragraph{Imitation Learning (IL).}
IL is another common tool for such problems. However, classical IL algorithms like behavioral cloning (BC)~\citep{pomerleau1991efficient,bain1995framework,schaal1999imitation,ross2011reduction}, inverse reinforcement learning (IRL)~\citep{abbeel2004apprenticeship,ramachandran2007bayesian,ziebart2008maximum}, and adversarial learning~\citep{ho2016generative,finn2016connection,fu2018learning} primarily focus on recovering the exact policy of the demonstrations, which may result in poor performance when given imperfect demonstrations.
A recent line of work on learning from suboptimal demonstrations offers a possible route to learn a policy that outperforms the demonstrations. However, they either require various types of manual supervision, such as rankings~\citep{brown2019extrapolating,zhang2021confidence}, weights of demonstrations~\citep{wu2019imitation,cao2021learning}, or have additional requirements on the environments~\citep{brown2020better,chen2021learning}, demonstrations~\citep{tangkaratt2020variational,tangkaratt2021robust}, or the training process~\citep{novoseller2020dueling}. Moreover, none of them can provide results about the stability of the learned policy. Another line of work learns certificates from demonstrations~\citep{ravanbakhsh2019learning,robey2020learning,chou2020uncertainty,boffi2021learning}, but they need additional assumptions such as known dynamics, perfect demonstrations, or the ability to query the demonstrator. In contrast, we leverage the CLF as natural guidance for the exploration process and do not require additional supervision or assumptions about the demonstrations or the environment to learn a stabilizing policy.

\section{Problem Setting and Preliminaries}\label{sec:preliminary}

We consider a discrete-time unknown dynamical system 
\begin{equation}\label{eq:dynamics}
    x(t+1)=h(x(t),u(t)),
\end{equation}
where $x(t)\in\mathcal{X}\subseteq\mathbb{R}^{n_x}$ represents the state at time step $t$, $u(t)\in\mathcal{U}\subseteq\mathbb{R}^{n_u}$ denotes the control input at time step $t$, and $h: \mathcal{X}\times\mathcal{U}\rightarrow\mathcal{X}$ is the \emph{unknown} dynamics. We assume the state space $\mathcal{X}$ to be compact and $h$ is Lipschitz continuous in both $(x, u)$ with constant $L_h>0$ following \cite{berkenkamp2017safe}. Our objective is to find a control policy $u(\cdot)=\pi(x(\cdot))$, where $\pi: \mathcal{X}\rightarrow\mathcal{U}$, such that from initial states $x(0)\in\mathcal{X}_0$, under the policy $\pi$, the closed-loop system asymptotically stabilizes at a goal $\xg \in \mathcal X$. In other words, $\forall x(t)$ starting from $x(0)\in\mathcal{X}_0$ and satisfying \Cref{eq:dynamics} with $u(t)=\pi(x(t))$, we have $\lim\limits_{t\rightarrow\infty}\left\|x(t)-\xg\right\|=0$.


We assume that we are given a set of $N_{D^0}$ demonstration transitions $\mathcal{D}^0=\{(x_i(t),u_i(t),x_i(t+1))\}_{i=1}^{N_{D^0}}$ generated by a demonstrator policy $\pi_d$ starting from states $x_i(0)\in\mathcal{X}_0$. In contrast to the assumption of stabilizing demonstrators in classical IL works~\citep{ho2016generative}, our demonstrations may not be generated by a stabilizing controller. 


Lyapunov theory is widely used to prove the stability of control systems, and CLFs offer further guidance for controller synthesis by defining a set of stabilizing control inputs at a given point in the state space. Following \cite{grune2017nonlinear}, we provide the definition of CLF for discrete system \eqref{eq:dynamics}.

\begin{definition}
    Consider the system \eqref{eq:dynamics} and a goal point $\xg$, and let $\mathcal{G}\subseteq\mathcal{X}$ be a subset of the state space such that $\xg\in\mathcal{G}$. A function $V:\mathcal{G}\rightarrow\mathbb{R}$ is called a CLF on $\mathcal{G}$ if there exists functions $\underline\alpha,\bar\alpha\in\mathcal{K}_\infty$\footnote{A function $\alpha: [0,+\infty)\rightarrow[0,+\infty)$ is said to be class-$\mathcal{K}_\infty$ if $\alpha$ is continuous, strictly increasing with $\alpha(0)=0$, and $\lim\limits_{s\rightarrow+\infty}\alpha(s)=+\infty$.} and a constant $\lambda\in(0,1)$ such that the following hold:
    \begin{subequations}\label{eq:clf-cond}
        \begin{align}
            & \underline\alpha(\|x-\xg\|)\leq V(x)\leq\bar\alpha(\|x-\xg\|) \label{eq:clf-pos-cond}\\
            & \inf_{u\in\mathcal U} V\left(h(x,u)\right)\leq\lambda V(x)\label{eq:clf-decrease-cond}
        \end{align}
    \end{subequations}
\end{definition}

The set of input $\mathcal{K}(x)=\left\{u\in\mathcal U\; |\; V\left(h(x,u)\right)\leq\lambda V(x)\right\}$ is called \textit{stablizing control inputs}. It is a standard result that if $\mathcal{G}$ is forward invariant\footnote{A set $\mathcal{G}$ is forward invariant if $x(0)\in\mathcal{G}\implies x(t)\in\mathcal{G}$ for all $t>0$.} and the goal point $\xg\in\mathcal{G}\subseteq\mathcal{X}$, then starting from initial set $\mathcal{X}_0\subseteq\mathcal{G}$, any control input $u \in \mathcal{K}$ will make the closed-loop system asymptotically stable at $\xg$ \citep{grune2017nonlinear}. The details are provided at \Cref{sec:app-clf}.

\section{{\algo} Algorithm}\label{sec:algo}

\noindent\textbf{Notation: }
Let $\tau$ be the current iteration step in our algorithm. A \emph{dataset} $\mathcal{D}^\tau$ is a set of $N_{D^\tau}$ transitions collected by the current iteration step. Recall that $\mathcal{D}^0$ is the set of given demonstrations. Let $\mathcal{D}_x^\tau\subset\mathcal{X}$ be the projection of $\mathcal{D}^\tau$ on the first state component of $\mathcal D^\tau$ defined as $\mathcal D_x^\tau = \{x \;|\; (x, \cdot, \cdot) \in \mathcal D^\tau\}$. A \emph{trusted tunnel} $\mathcal{H}^\tau$ is defined as the set of states at most $\gamma > 0$ distance away from the dataset $\mathcal{D}^\tau_x$, \ie, $\mathcal{H}^\tau =\{x\; |\; \exists x_i\in\mathcal{D}_x^\tau,\|x-x_i\|\leq\gamma\}$. 



We first outline our algorithm \textbf{LY}apunov-\textbf{G}uided \textbf{E}xploration (\algo), which learns to stabilize high-dimensional unknown systems, then provide a step-by-step explanation (see \Cref{sec:app-analysis} for detailed theoretical analysis).
Given a dataset of imperfect demonstrations $\mathcal{D}^0$ (which may not contain $\xg$), we firstly employ IL to learn an initial controller $\pi_\mathrm{init}$ (which may be an unstable controller). Then, during each iteration $\tau$,  we learn a CLF $V^\tau_\theta$ and the corresponding controller $\pi^\tau_\phi$ using samples from $\mathcal D^\tau_x$, and use the learned controller $\pi^\tau_\phi$ to generate closed-loop trajectories as additional data added to $\mathcal{D}^\tau$ to obtain $\mathcal D^{\tau+1}$. 
At each iteration, the learned CLF $V^\tau_\theta$ ensures that the newly collected trajectories get closer to $\xg$, as each $V^\tau_\theta$ is designed to reach the global minimum at $\xg$.
In this way, the trusted tunnel $\mathcal{H}^\tau$ grows and includes states closer and closer to $\xg$ with more iterations. 
Upon convergence, \algo\ returns a stabilizing controller $\pi^*$ that can be trusted within the converged trusted tunnel $\mathcal{H}^*$, where $\mathcal{H}^*$ contains all trajectories starting from $\mathcal{X}_0$.

\noindent\textbf{Learning from Demonstrations: } 
At iteration $0$, our approach starts with learning an initial policy from imperfect and potentially unstable demonstrations $\mathcal{D}^0$. We use existing IL methods (\eg, BC~\citep{bain1995framework}) to learn the initial policy $\pi_\mathrm{init}(\cdot)$. Since the IL algorithms directly recover the behavior of the demonstrations, the initial policy could lead to unstable behaviors. 

\noindent\textbf{Learning Local Dynamics: }
At iteration $\tau$, we learn a discrete NN approximation of the dynamics $\hat h^\tau_{\psi}(x,u): \mathcal{X}\times\mathcal{U}\rightarrow\mathcal{X}$ parameterized by $\psi$ using the transition dataset $\mathcal{D}^\tau$.
We train $\hat{h}^\tau_\psi$ by minimizing the mean square error loss from transitions sampled from the dataset $\mathcal{D}^\tau$ using Adam~\citep{kingma2014adam}. Let $\omega>0$ be the maximum error of the learned dynamics on the training data: $\|\hat{h}_\psi^\tau(x_i(t),u_i(t))-x_i(t+1)\|\leq\omega$ for all $(x_i(t),u_i(t),x_i(t+1))\in\mathcal{D}^\tau$. 


\noindent\textbf{Learning CLF and Controller: } 
After learning the local dynamics, we jointly learn the CLF and the control policy. Let the learned CLF in $\tau$-th iteration be 
$V^\tau_\theta(x)=x^\top S^\top S x+p_\mathrm{NN}(x)^\top p_\mathrm{NN}(x)$, 
where $S\in\mathbb{R}^{n_x\times n_x}$ is a matrix of parameters, $p_\mathrm{NN}:\mathbb{R}^{n_x}\rightarrow\mathbb{R}^{n_x}$ denotes an NN, and $\theta$ encompasses all parameters including $S$ and the ones in $p_\mathrm{NN}$. 
Clearly, $V^\tau_\theta(x)$ is positive by construction. The first term in $V^\tau_\theta$ models a quadratic function, which is commonly used to construct Lyapunov functions for linear systems. Here we use it to introduce a quadratic prior to the learned CLF. The second term models the CLF's non-quadratic residue.
We also parameterize the controller with an NN, \ie, $\pi^\tau_\phi:\mathcal{X}\rightarrow\mathcal{U}$ with parameters $\phi$. 
The learned controller aims to direct system trajectories toward the goal, guided by the learned CLF. Additionally, since the learned dynamics $\hat h^\tau_\psi$ may be invalid outside the trusted tunnel $\mathcal{H}^\tau$, $\pi^\tau_\phi$ should not drive the system too far from $\mathcal{H}^\tau$ within the simulation horizon. This can be achieved by penalizing the distance of the control inputs in consecutive iterations.
Overall, $V^\tau_\theta$ and $\pi^\tau_\phi$ can be synthesized by solving the following optimization problem:
\begin{subequations}\label{eq:clf-opt}
    \begin{align}
    \min_\phi \qquad & \mathbb{E}_{x\in\mathcal{H}^\tau}\|\pi^\tau_\phi(x)-\pi^{\tau-1}_\phi(x)\|\label{eq:ctrl-obj}\\
    \mathrm{s.t.}\qquad & V^\tau_\theta(\xg)=0, \qquad V^\tau_\theta(x)>0,\quad \forall x\in\mathcal{X}\setminus\xg \\
    & V^\tau_\theta\left(\hat{h}^\tau_\psi(x,\pi^\tau_\phi(x))\right) \leq\lambda V^\tau_\theta(x),\quad \forall x\in\mathcal{H}^\tau\label{eq:clf-nn-decrease-cond}
    \end{align}
\end{subequations}
where $\lambda\in (0, 1)$. 
We approximate the solution of problem~\eqref{eq:clf-opt} by self-supervised learning with loss $\mathcal{L}^\tau=\mathcal{L}_\mathrm{CLF}^\tau+\eta_\mathrm{ctrl}\mathcal{L}_\mathrm{ctrl}^\tau$, where $\mathcal{L}^\tau_\mathrm{CLF}$ and $\mathcal{L}^\tau_\mathrm{ctrl}$ correspond to the constraints and the objective of Problem \eqref{eq:clf-opt}, respectively, $\eta_\mathrm{ctrl}>0$ is a hyper-parameter that balances the weights of two losses, and
{\small\begin{align}
    \mathcal{L}^\tau_\mathrm{CLF}=&V^\tau_\theta(\xg)^2+\frac{1}{N}\sum_{x\in\mathcal{X}\setminus\xg}\left[\nu-V^\tau_\theta(x)\right]^+
    + \frac{\eta_\mathrm{pos}}{N}\sum_{x\in\mathcal{D}^\tau_x}\left[\epsilon+V^\tau_\theta\left(\hat{h}^\tau_\psi(x,\pi^\tau_\phi(x))\right)-\lambda V^\tau_\theta(x)\right]^+\label{eq:loss-clf}\\
    \mathcal{L}^\tau_\mathrm{ctrl}=&\frac{1}{N}\sum_{x\in\mathcal{D}^\tau_x}\|\pi^\tau_\phi(x)-\pi^{\tau-1}_\phi(x)\|^2,\label{eq:loss-ctrl}
\end{align}}\normalsize
where $[\cdot]^+=\max(\cdot, 0)$, $\eta_\mathrm{pos}>0$ is a hyper-parameter that balances each term in the loss, and $\epsilon>0$ is a hyper-parameter that ensures that the learned CLF $V^\tau_\theta$ satisfies condition \eqref{eq:clf-decrease-cond} even if the learned dynamics $\hat{h}^\tau_\psi$ has errors (Studied in \Cref{sec:ablation}).
In practice, it is hard to enforce $V^\tau_\theta(\xg)$ to be exactly $0$ using gradient-based methods. Therefore, in loss \eqref{eq:loss-clf} and \eqref{eq:loss-ctrl}, we instead train the NNs to make $V^\tau_\theta(\xg)<\nu$, and $V^\tau_\theta(x)\geq\nu$, for all $x\in\mathcal{X}\setminus\xg$, where $\nu$ is a small positive number.

\begin{figure*}[t]
    \centering
    \subfigure[Iteration 0]{\includegraphics[width=.245\textwidth]{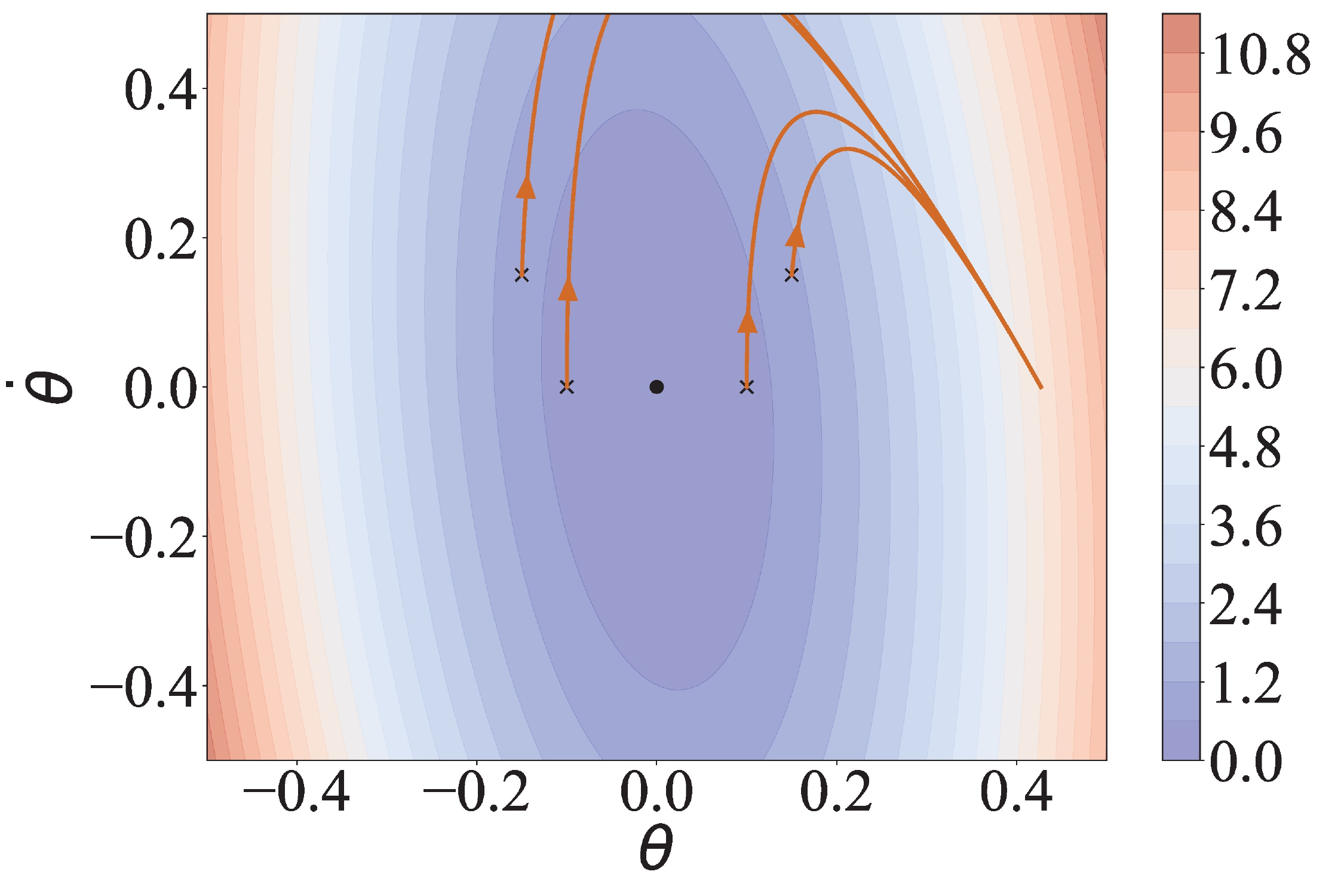}\label{fig:exploration-0}}
    \subfigure[Iteration 2]{\includegraphics[width=.245\textwidth]{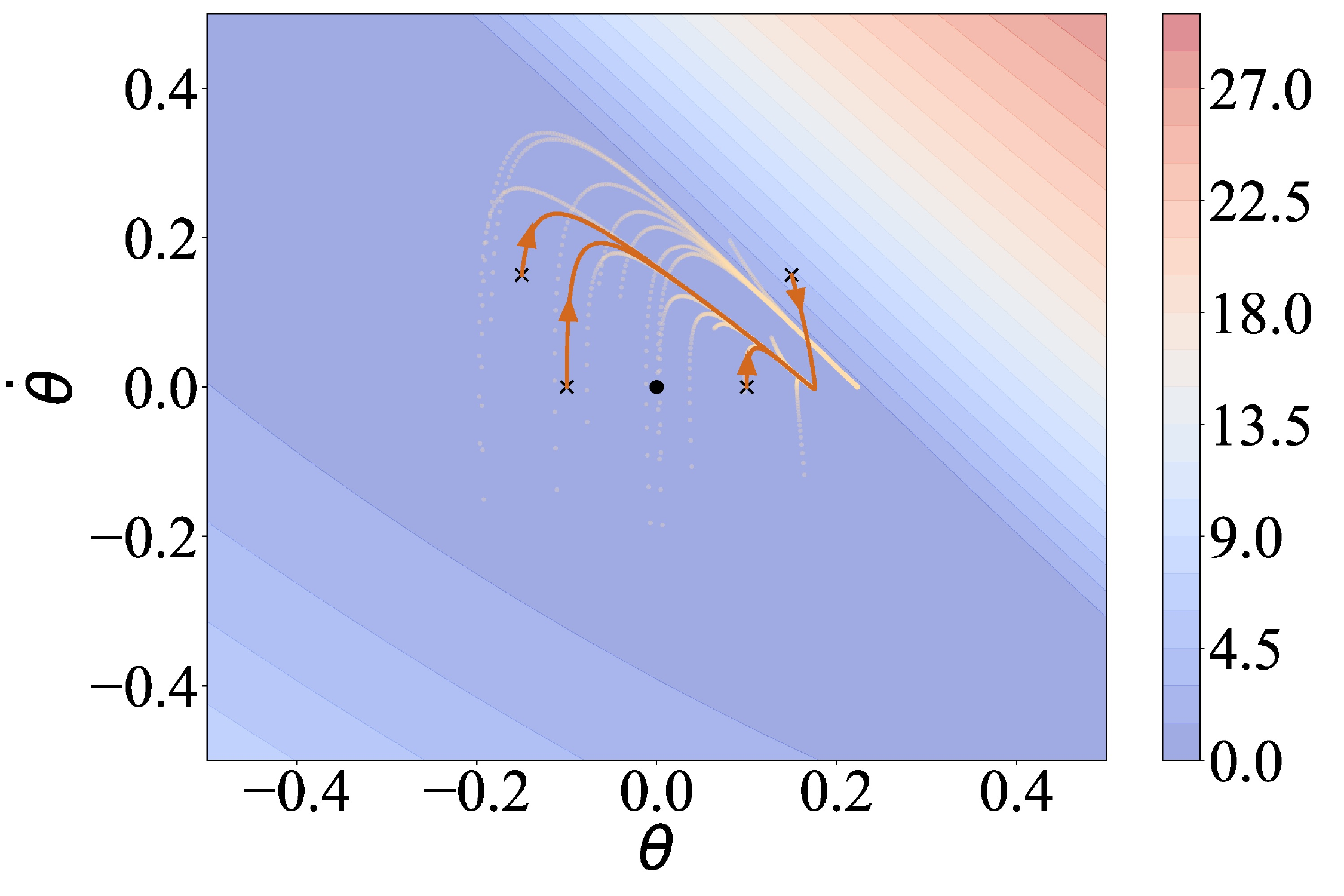}\label{fig:exploration-2}}
    \subfigure[Iteration 6]{\includegraphics[width=.245\textwidth]{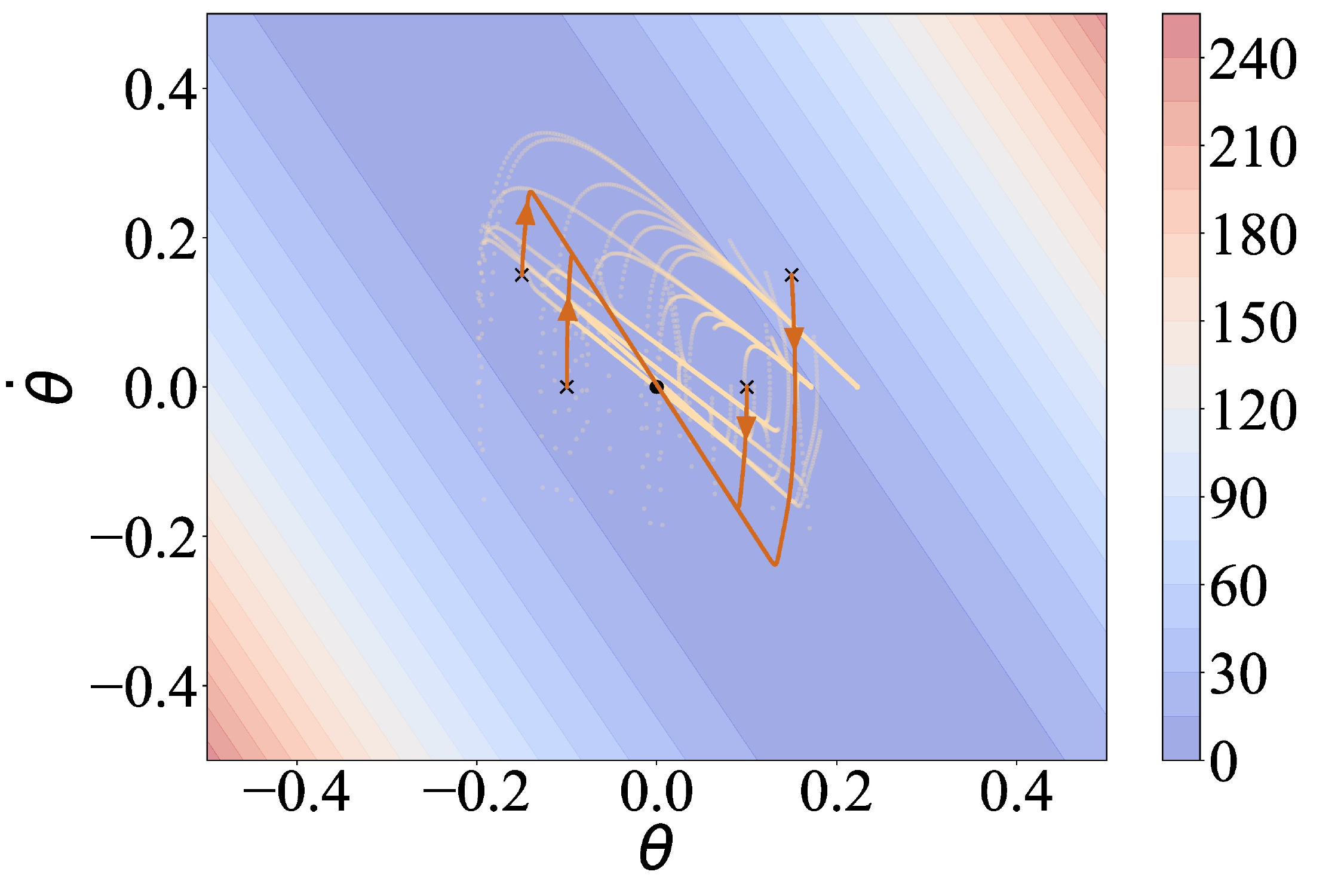}\label{fig:exploration-6}}
    \subfigure[Iteration 12]{\includegraphics[width=.245\textwidth]{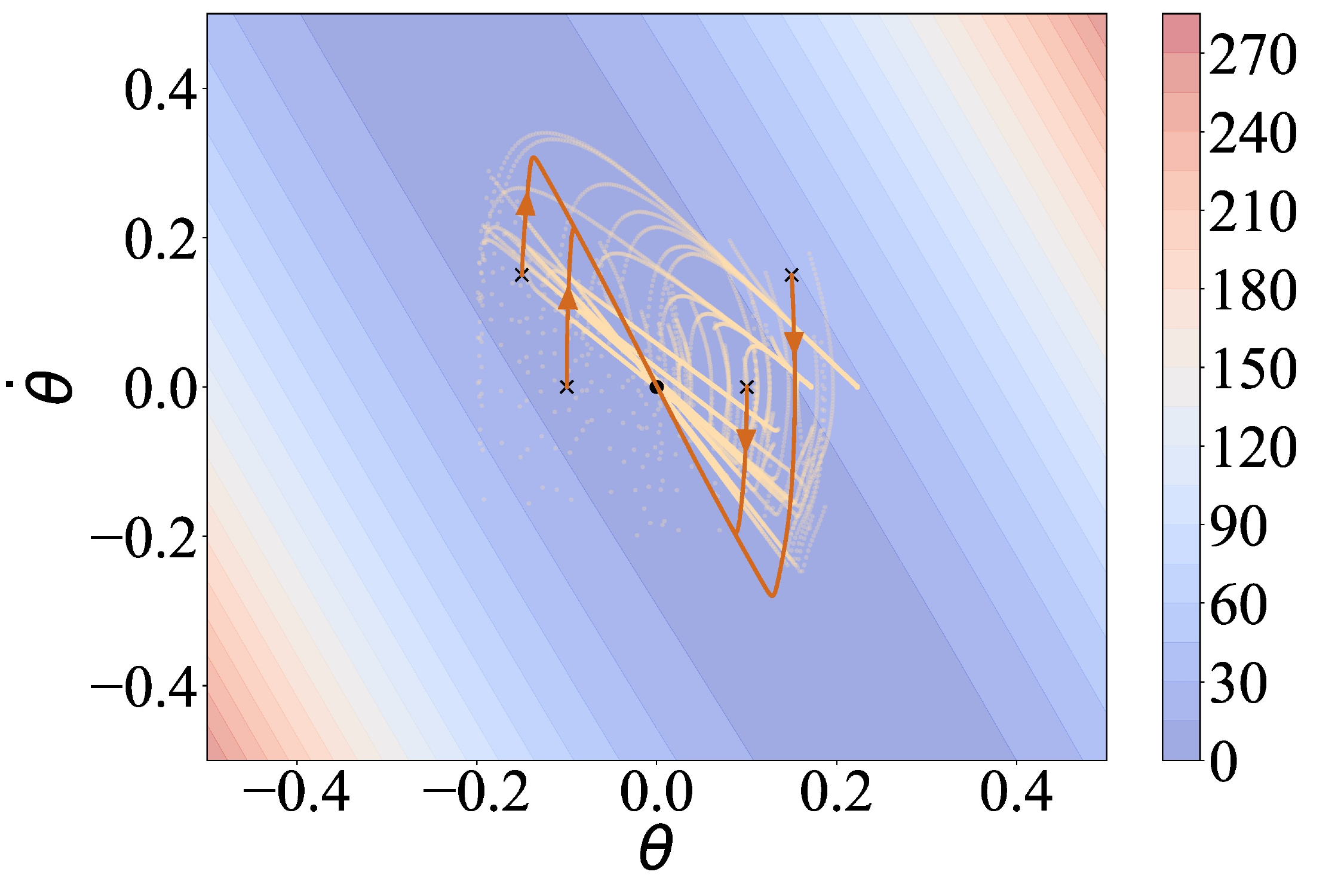}\label{fig:exploration-12}}
    \caption{Trajectories generated by \algo\ in different iterations in inverted pendulum environment. The counters show the learned CLF. The orange trajectories are generated by the learned controller in the current iteration. The light orange dots are the demonstrations generated in previous iterations, which also indicate the trusted tunnel $\mathcal{H}$. The black dot is the goal.}
    \label{fig:exploration}
\end{figure*}

\noindent\textbf{Exploration: }
After learning $\hat h^\tau_\psi$, $V^\tau_\theta$, and $\pi^\tau_\phi$ in the $\tau$-th iteration, we use the current controller $\pi_\phi^\tau$ starting from $x_0\in\mathcal{H}^\tau$ to collect $N_D$ more transitions $\Delta \mathcal D^\tau = \{(x_i(t),\pi^\tau_\phi(x_i(t)),x_i(t+1))\}_{i=1}^{N_D}$ and augment the dataset $\mathcal{D}^{\tau+1}=\mathcal{D}^\tau\cup\Delta \mathcal D^\tau$. During the exploration, the controller drives the system to states with lower CLF values. Consequently, by collecting more trajectories with the learned policy $\pi^\tau_\phi$, we expand the trusted tunnel $\mathcal{H}^\tau$ to states with lower CLF values. As the global minimum of the CLF is at $\xg$, the system trajectories used to construct the trusted tunnel keep getting closer to the goal in each iteration. 
The detailed algorithm and the convergence result are provided in \Cref{sec:app-analysis}.

We illustrate the \algo\ process in an inverted pendulum environment in \Cref{fig:exploration}. In \Cref{fig:exploration-0}, since the demonstrations are imperfect, our initial controller $\pi_\mathrm{init}$ cannot reach the goal. \Cref{fig:exploration-2} and \Cref{fig:exploration-6} demonstrate that after several iterations, the trusted tunnel $\mathcal{H}^\tau$ (the region around the light orange dots $\mathcal{D}^\tau_x$) is expanded towards the goal, and the closed-loop system progressively approaches the goal. Upon convergence (\Cref{fig:exploration-12}), our controller stabilizes the system at the goal. 

\section{Experiments}\label{sec:experiments}

We conduct experiments in six environments including Inverted Pendulum, Cart Pole, Cart II Pole, Neural Lander~\citep{shi2019neural}, and the F-16 jet~\citep{heidlauf2018verification} with two tasks: ground collision avoidance (GCA) and tracking. 
To simulate the imperfect demonstrations, We collect imperfect and potentially unstable demonstrations for each environment using nominal controllers such as LQR~\citep{kwakernaak1969linear} for Inverted Pendulum, PID~\citep{bennett1996brief} for Neural Lander and the F-16, and RL controllers for Cart Pole and Cart II Pole. In the first four environments, we collect $20$ trajectories as demonstrations, and in the two F-16 environments, we collect $40$ trajectories. We aim to answer the following questions in the experiments: 
1) How does \algo\ compare with other algorithms for the case of stabilizing the system at goal? 2) What is the sampling efficiency of LYGE as compared to other baseline methods? 3) Can \algo\ be used for systems with high dimensions?
We provide additional implementation details and more results in \Cref{sec:app-experiments}. 

\subsection{Baselines}
We compare \algo\ with the most relevant works in our problem setting including RL algorithm PPO~\citep{schulman2017proximal}, standard IL algorithm AIRL~\citep{fu2018learning}, and algorithms of IL from suboptimal demonstrations D-REX~\citep{brown2020better} and SSRR~\citep{chen2021learning}. For PPO, we hand-craft reward functions following standard practices in reward function design~\citep{gym} for stabilization problems. For AIRL, D-REX, and SSRR, we let them learn directly from the demonstrations. 
For a fair comparison, we initialize all the algorithms with the BC policy. 
Compared with these baselines, our algorithm has one additional assumption that we know the desired goal point. However, we believe that the comparison is still fair because we do not need the reward function or optimal demonstrations. While \algo\ needs more information than D-REX and SSRR, the performance increment from \algo\ is large enough that it is worth the additional information. 

We also design two other baselines, namely, CLF-sparse and CLF-dense, to show the efficacy of the Lyapunov-guided exploration compared with learning the dynamics model from random samples \citep{dai2021lyapunov,zhou2022neural}. 
These methods require the stronger assumption of being able to sample from arbitrary states in the state space, which is unrealistic when performing experiments outside of simulations. 
For both algorithms, we follow the same training process as \algo, but instead of collecting samples of transitions by applying Lyapunov-guided exploration, we directly sample states and actions from the entire state-action space to obtain the training set of the dynamics. We note that CLF-sparse uses the same number of samples as \algo, while CLF-dense uses the same number of samples as the RL and IL algorithms, which is much more than \algo\ (see \Cref{tab:n-sample}). 

\subsection{Environments}

\noindent\textbf{Inverted Pendulum: }
Inverted Pendulum (Inv Pendulum) is a standard benchmark for testing control algorithms. 
To simulate imperfect demonstrations, the demonstration data is collected by an LQR controller with noise that leads to oscillation of the pendulum around a point away from the goal. 

\noindent\textbf{Cart Pole and Cart II Pole: }
Both environments are standard RL benchmarks introduced in Open AI Gym~\citep{gym}\footnote{Original names are InvertedPendulum and InvertedDoublePendulum in Mujoco environments}. 
We collect demonstrations using an RL policy that has not fully converged, which makes the cart pole oscillate at a location away from the goal. 

\noindent\textbf{Neural Lander: }
Neural lander~\citep{shi2019neural} is a widely used benchmark for systems with unknown disturbances. The state space has 6 dimensions including the 3-dimensional position and the 3-dimensional velocity, with 3-dimensional linear acceleration as the control input. The goal is to stabilize the neural lander at a user-defined point near the ground. The dynamics are modeled by a neural network trained to approximate the aerodynamics ground effect, which is highly nonlinear and unknown. We use a PID controller to collect demonstrations, which makes the neural lander oscillate and cannot reach the goal point because of the strong ground effect. 

\noindent\textbf{F-16: }
The F-16 model \citep{heidlauf2018verification,djeumou2021fly} is a high-fidelity fixed-wing fighter model, with 16D state space and 4D control inputs. The dynamics are complex and cannot be described as ODEs. Instead, the authors of the F-16 model provide many lookup tables to describe the dynamics.  We solve the two tasks discussed in the original papers: ground collision avoidance (GCA) and waypoint tracking. In GCA, the F-16 starts at an initial condition with the head pointing at the ground. The goal is to pull up the aircraft as soon as possible, avoid colliding with the ground, and fly at a height between $800$ ft and $1200$ ft. In the tracking task, the goal of the aircraft is to reach a user-defined waypoint.  The original model provides PID controllers. However, the original PID controller cannot pull up the aircraft early enough or cannot track the waypoint precisely. 

\begin{figure*}[t]
    \centering
    \subfigure[]{\includegraphics[width=.15\textwidth]{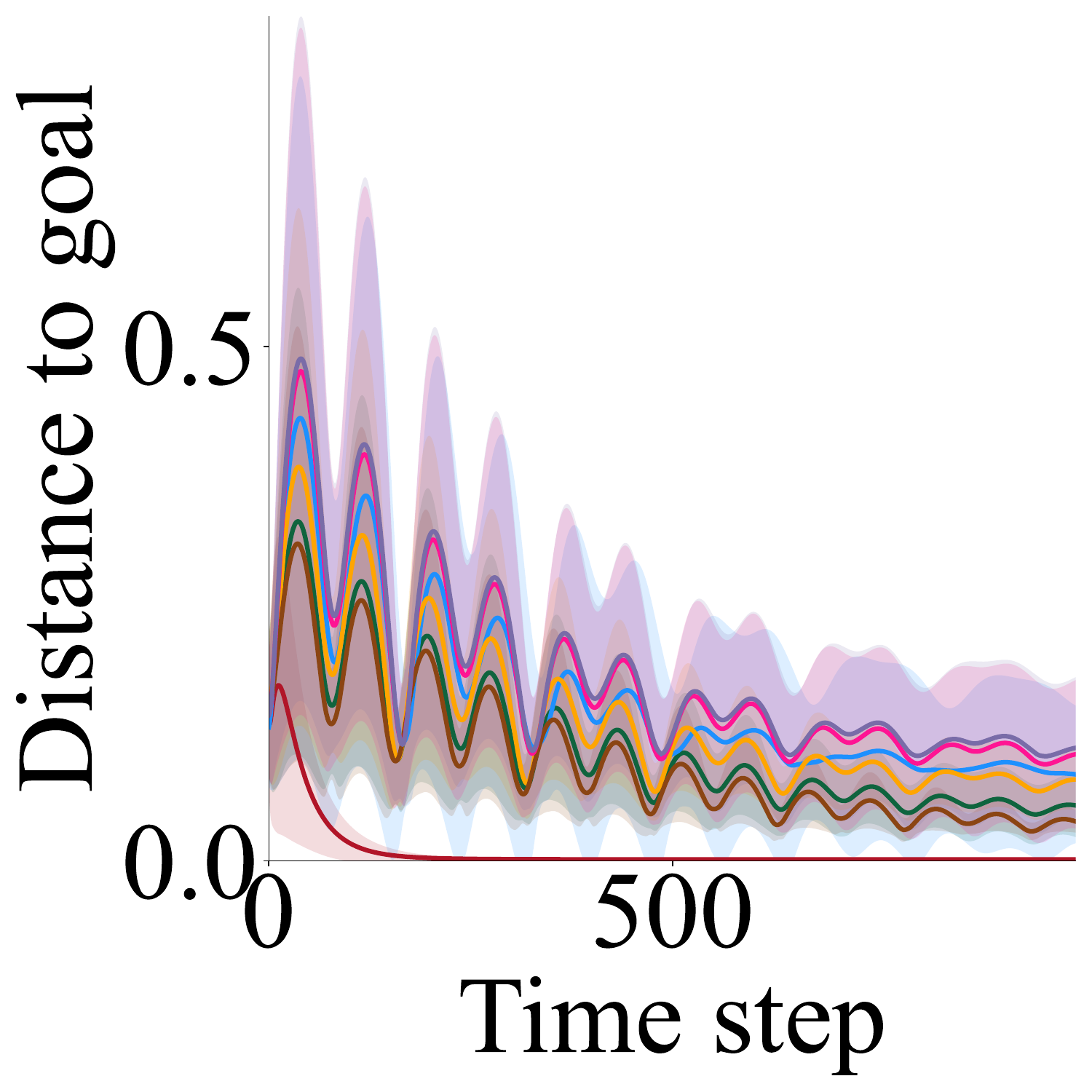}
    \label{fig:distance-invertedpendulum}}
    \subfigure[]{\includegraphics[width=.15\textwidth]{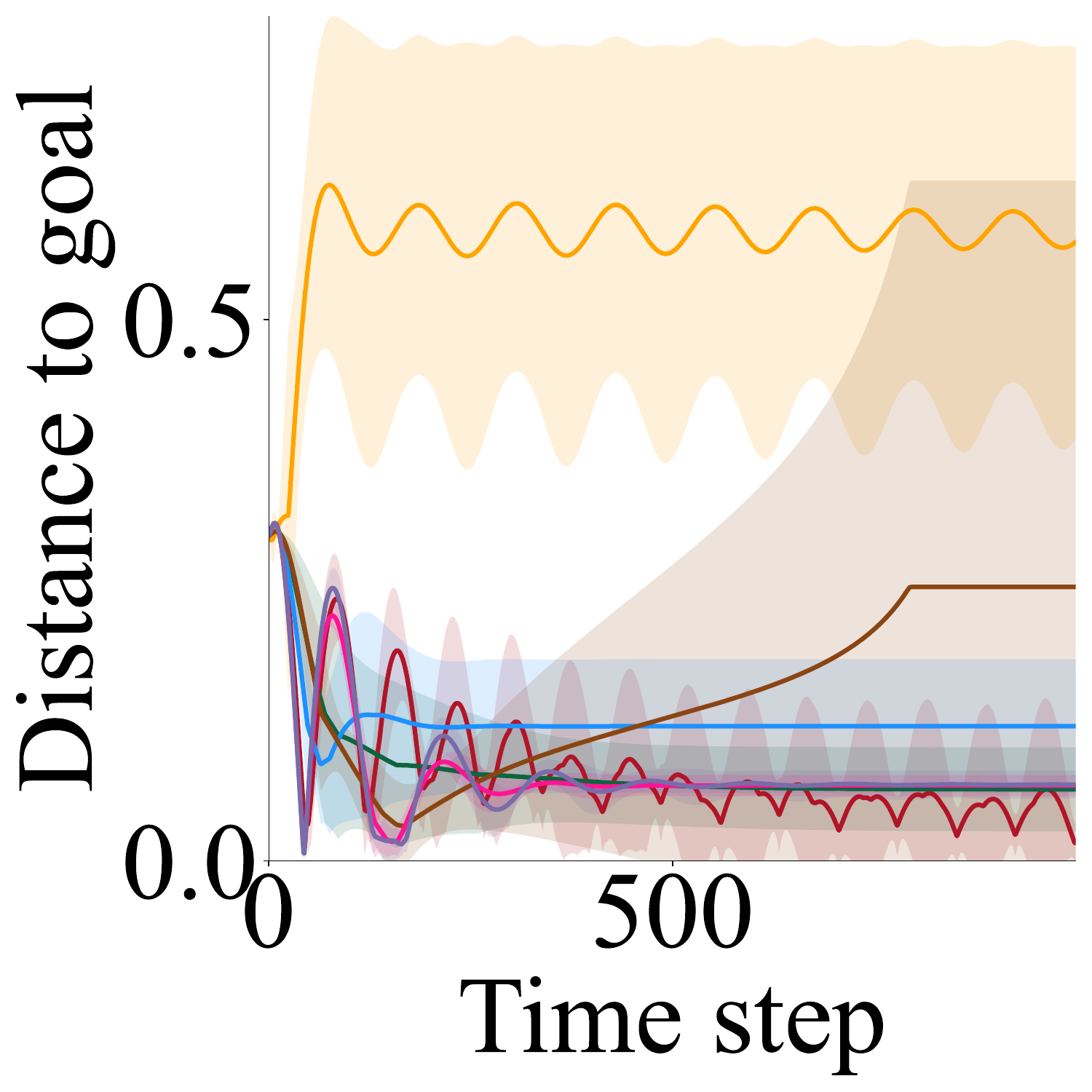}}
    \subfigure[]{\includegraphics[width=.15\textwidth]{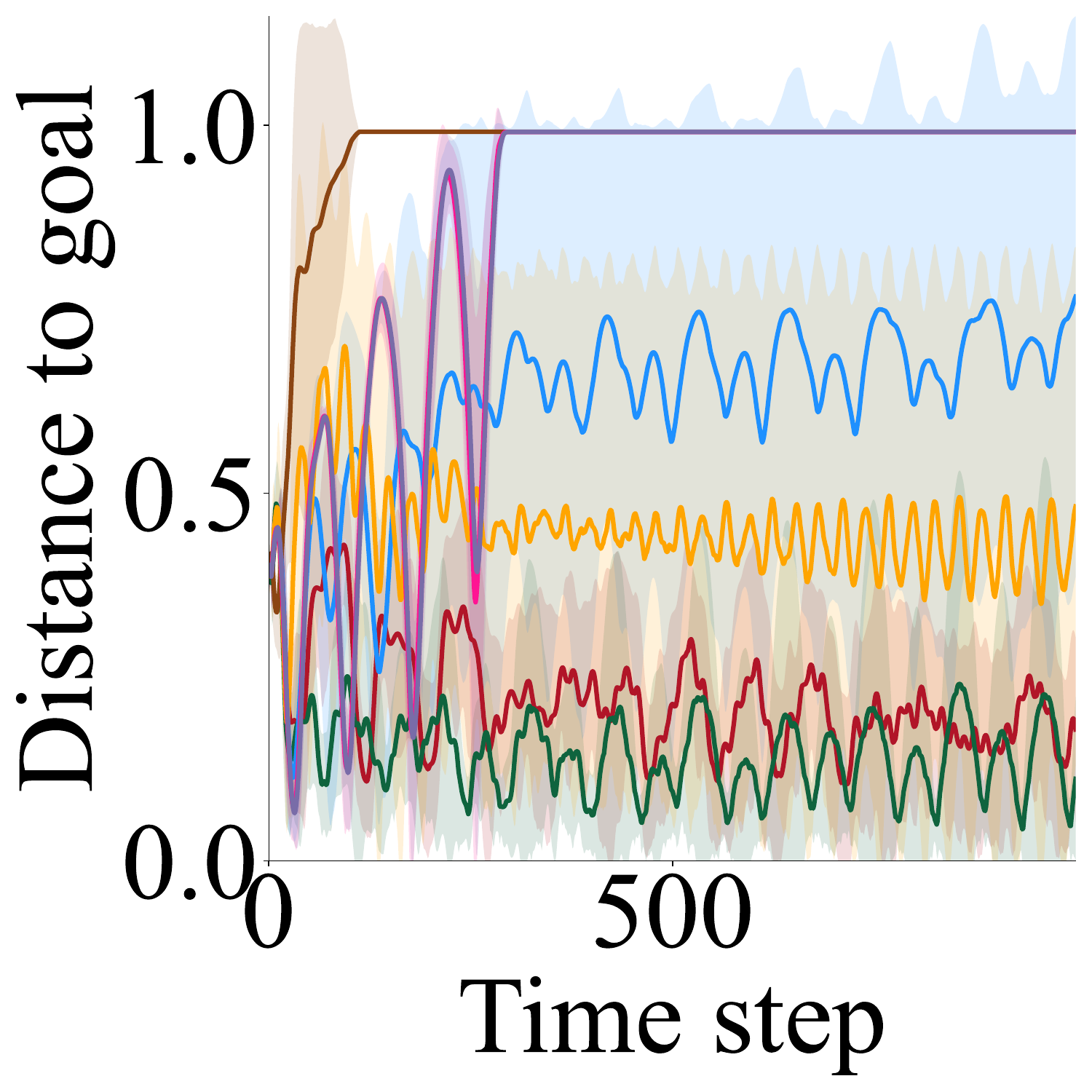}}
    \subfigure[]{\includegraphics[width=.15\textwidth]{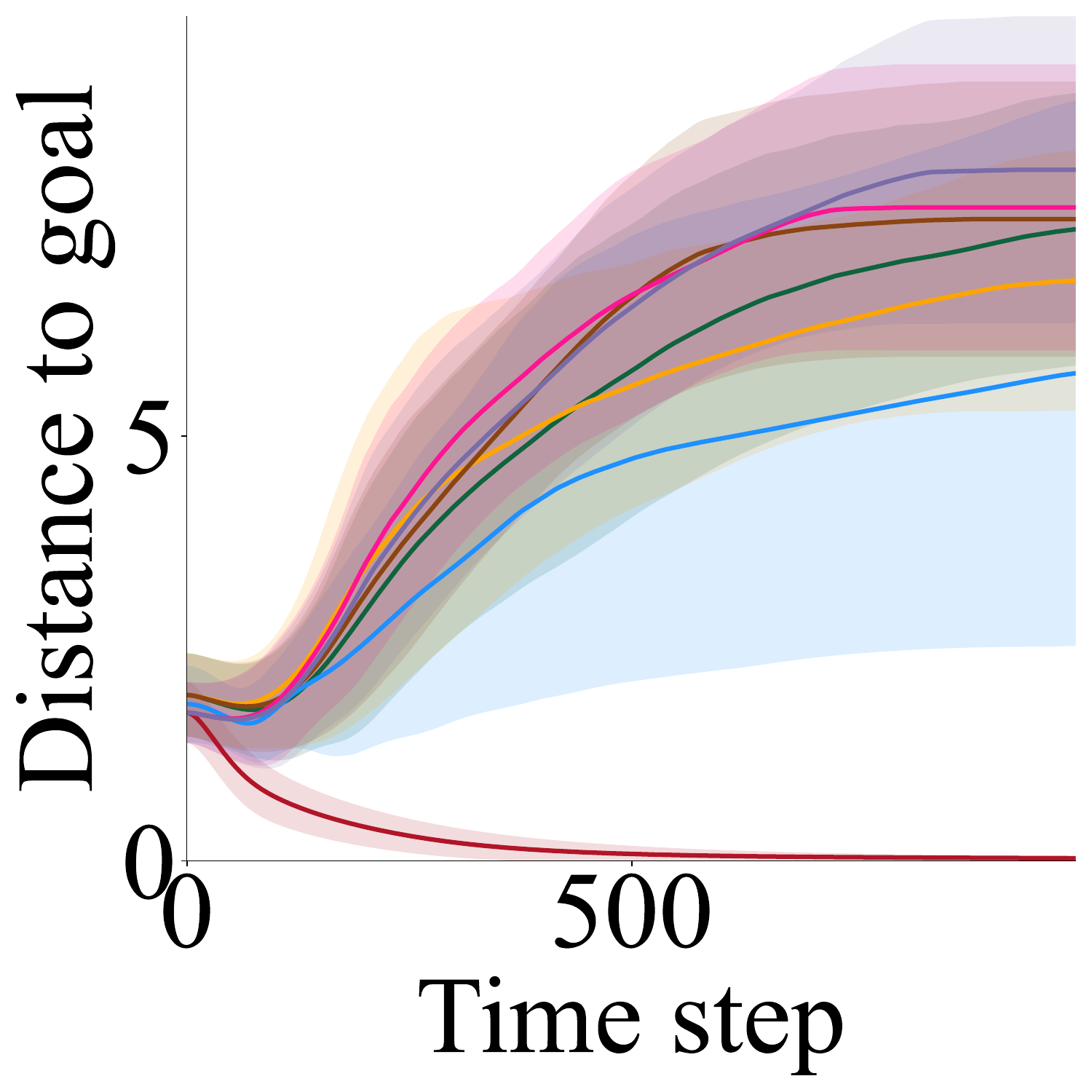}\label{fig:distance-neurallander}}
    \subfigure[]{\includegraphics[width=.15\textwidth]{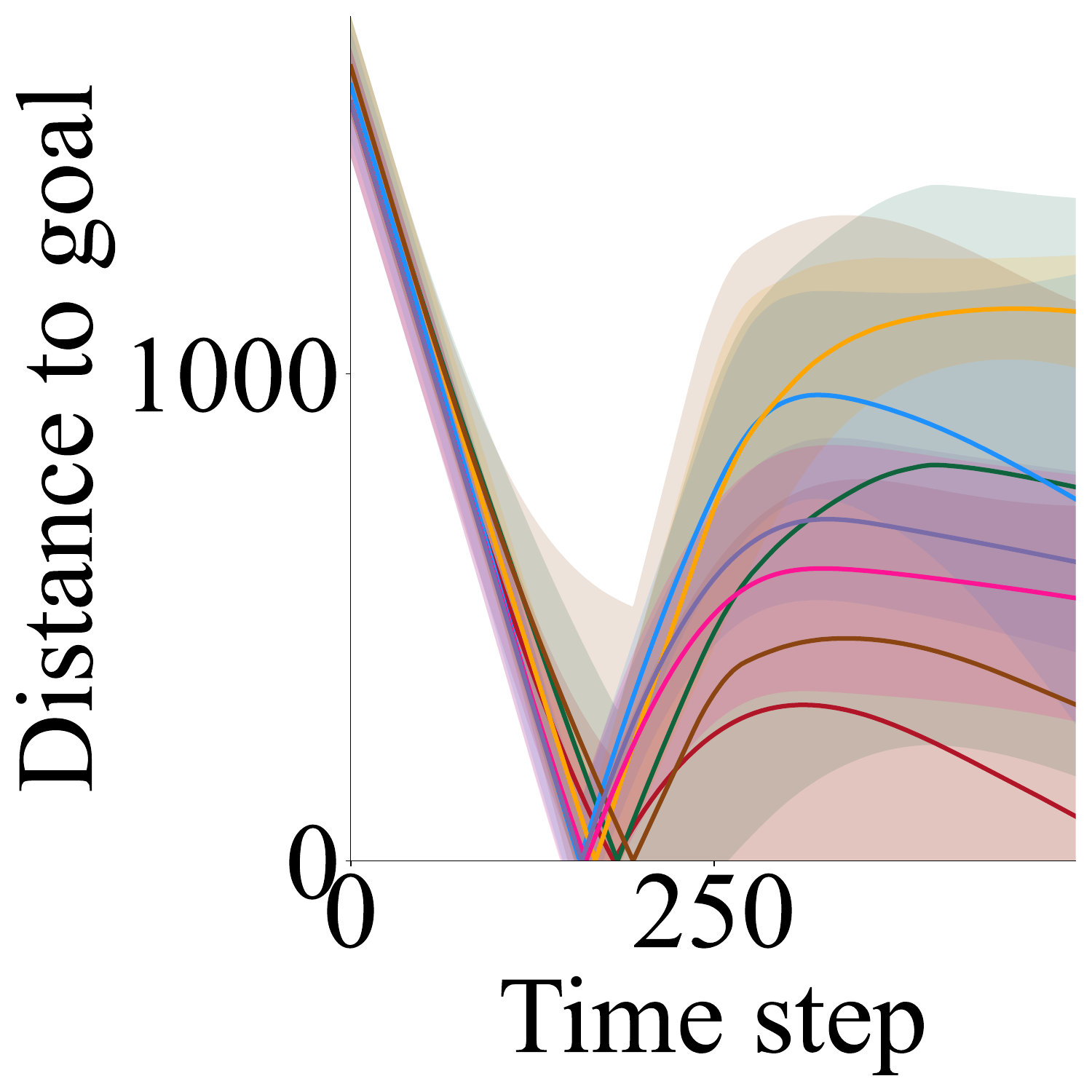}\label{fig:distance-f16gcas}}
    \subfigure[]{\includegraphics[width=.15\textwidth]{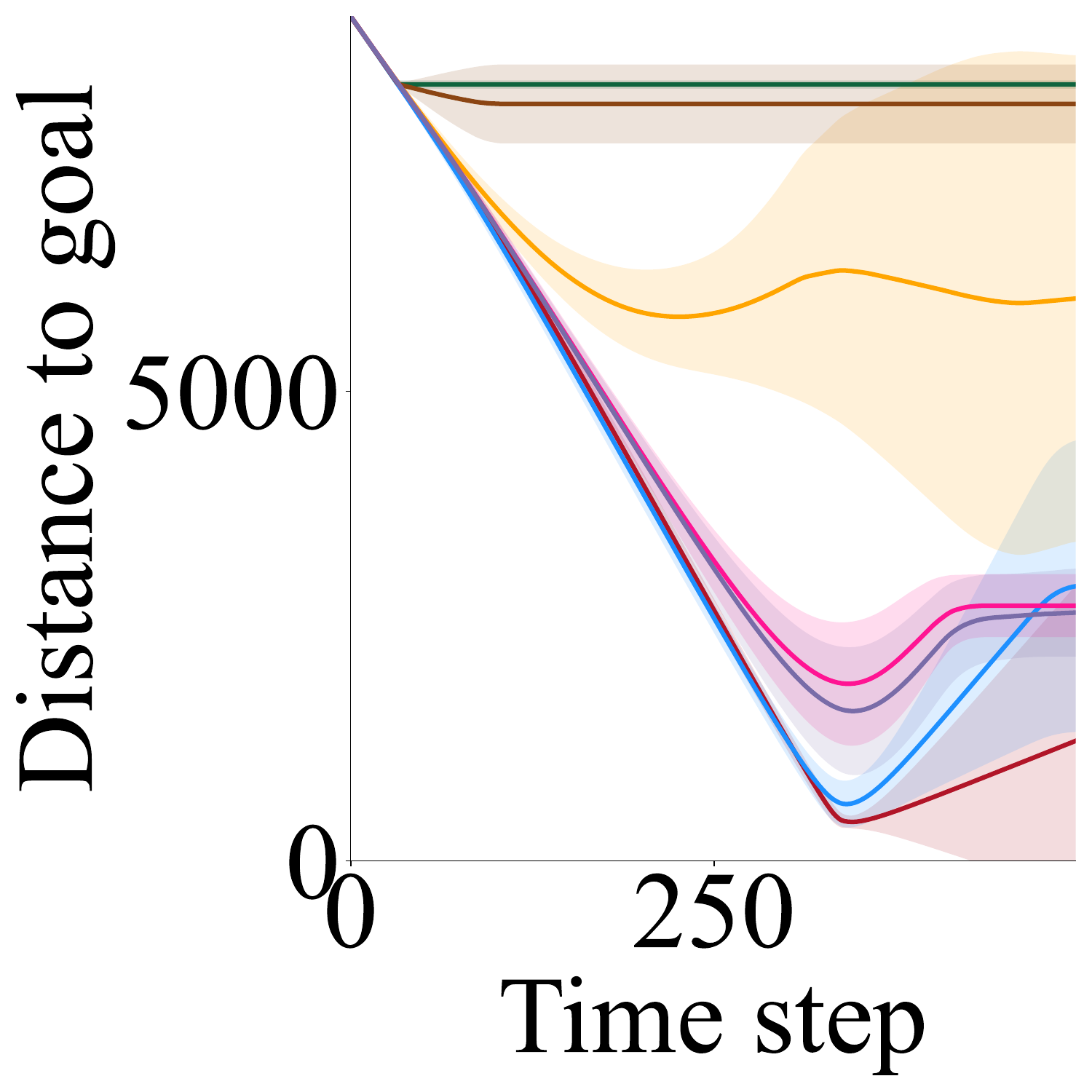}\label{fig:distance-f16tracking}}
    \includegraphics[width=.8\textwidth]{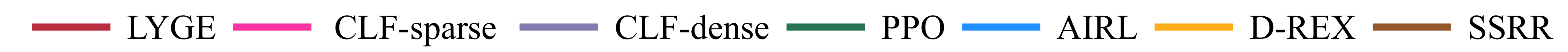}
    \caption{The distance to the goal w.r.t. time step of \algo\ and the baselines: (a) Inv Pendulum; (b) Cart Pole; (c) Cart II Pole; (d) Neural Lander; (e) F-16 GCA; (f) F-16 Tracking. The solid lines show the mean distance while the shaded regions show the standard deviation. Note that the curve corresponding to CLF-sparse almost overlaps with the curve corresponding to CLF-dense and so, the curve for CLF-sparse might not be visible in some of the plots.}
    \label{fig:distance-t}
\end{figure*}

\subsection{Results and Discussions}
\label{sec:experiment-results}

We train each algorithm in each environment $5$ times with different random seeds and test the converged controllers $20$ times each. In \Cref{fig:distance-t} we show the distance to the goal w.r.t. the simulation time steps. Note that we consider the goal height in the F-16 GCA environment and the goal position in the F-16 Tracking environment. We can observe that \algo\ achieves comparable or better results in terms of stabilizing the systems in all environments, especially in high-dimensional complex systems like Neural Lander and F-16. In Neural Lander, none of the baselines can reach the goal as they prioritize flying at a higher altitude to avoid collisions. In F-16 environments, the baselines either pull up the aircraft too late or have large tracking errors. 
\algo\, however, can finish these tasks perfectly. 

Specifically, compared with PPO, \algo\ achieves comparable results in simpler environments like Cart Pole and Cart II Pole, and behaves much better in complex environments like Neural Lander and F-16. 
This is due to the fact that PPO is a policy gradient method that approximates the solution of the Bellman equation, but getting an accurate approximation is hard for high-dimensional systems. 
In PPO the reward function can only describe ``where the goal is", but our learned CLF can explicitly tell the system ``how to reach the goal". 
Compared with AIRL, which learns the same policy as the demonstrations and cannot make improvements, \algo\ learns a policy that is much better than the demonstrations. Compared with ranking-guided algorithms D-REX and SSRR, \algo\ behaves better because ranking guidances provide less information than our CLF guidance, and are not designed to explicitly encode the objective of reaching the goal.
Compared with CLF-sparse and CLF-dense, \algo\ outperforms them because Lyapunov-guided exploration provides a more effective way to sample in the state space to learn the dynamics rather than random sampling. Although CLF-dense uses many more samples than CLF-sparse and \algo\, its performance does not significantly improve. This shows that na\"ively increasing the number of samples without guidance does little to improve the accuracy of the learned dynamics in high-dimensional spaces. 

In \Cref{tab:n-sample}, we show the number of samples used in the training. It is shown that \algo\ needs $68\%$ to $95\%$ fewer samples than other algorithms. This indicates that our Lyapunov-guided exploration explores only the necessary regions in the state space and thus improving the sample efficiency. 

\begin{table}[t]
    \centering
    \footnotesize
    \caption{Number of samples (k) used for training in different environments.}
    \begin{tabular}{c|cccccc}
        \toprule
        Algorithm & Inv Pendulum & Cart Pole & Cart II Pole & Neural Lander & F-16 GCA & F-16 Tracking \\
        \midrule
        \algo & 160 & 160 & 480 & 240 & 480 & 560 \\
        Baselines & 2000 & 500 & 5000 & 5000 & 2000 & 10000 \\
        \bottomrule
    \end{tabular}
    \label{tab:n-sample}
\end{table}

\begin{figure*}[t!]
    \centering
    \subfigure[]{\includegraphics[width=.19\textwidth]{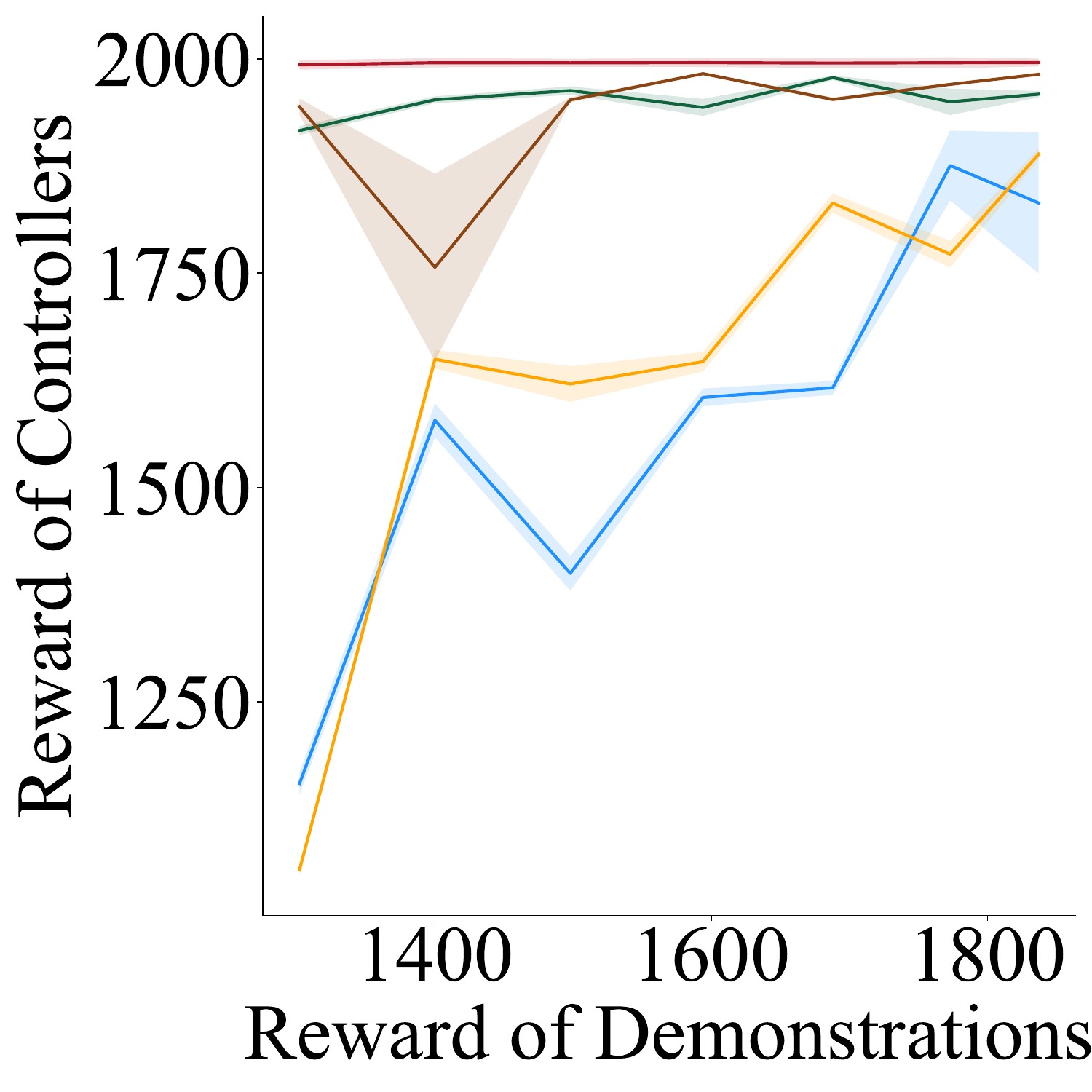}
    \label{fig:vary-opt}}
    \subfigure[]{\includegraphics[width=.19\textwidth]{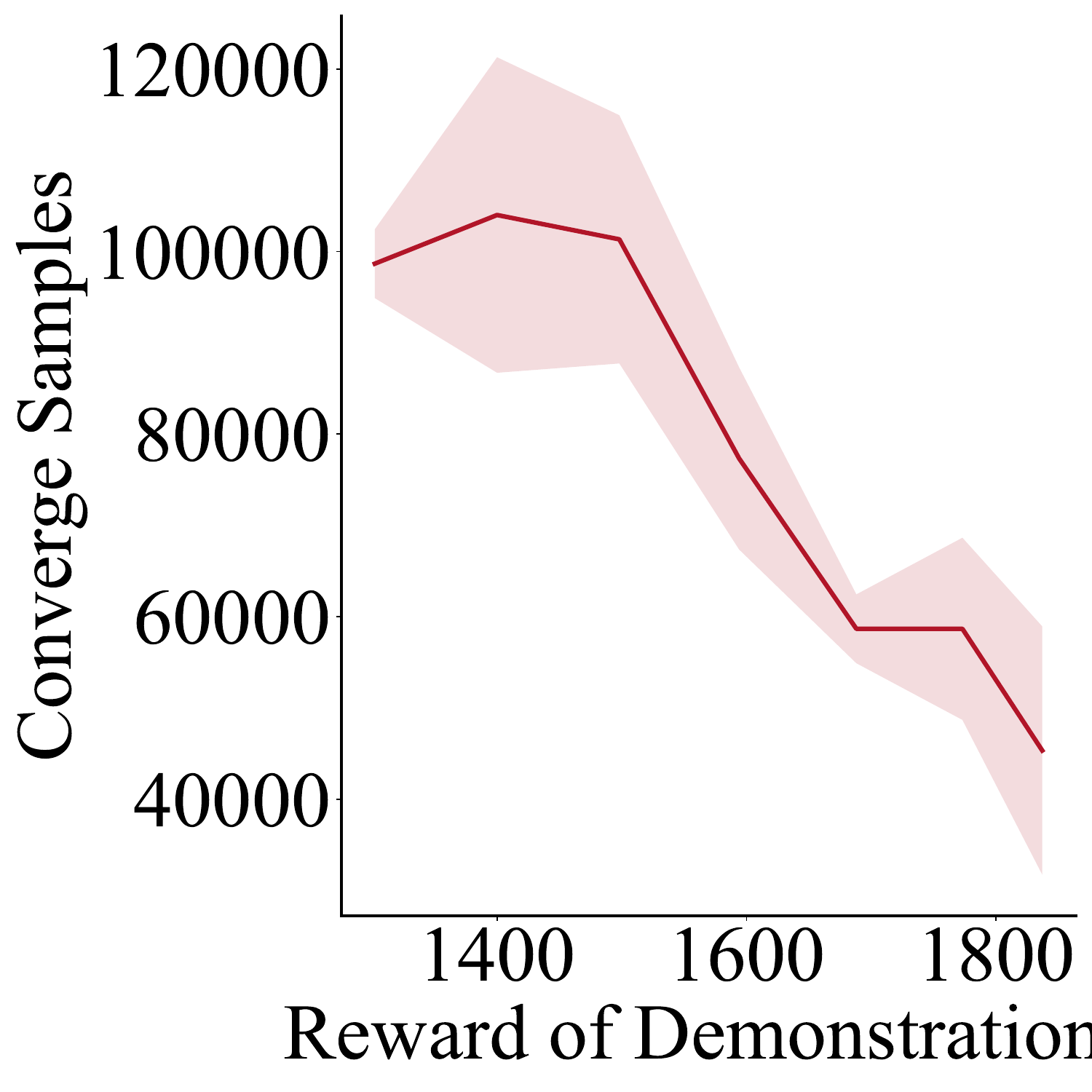}\label{fig:conv-time}}
    \subfigure[]{\includegraphics[width=.19\textwidth]{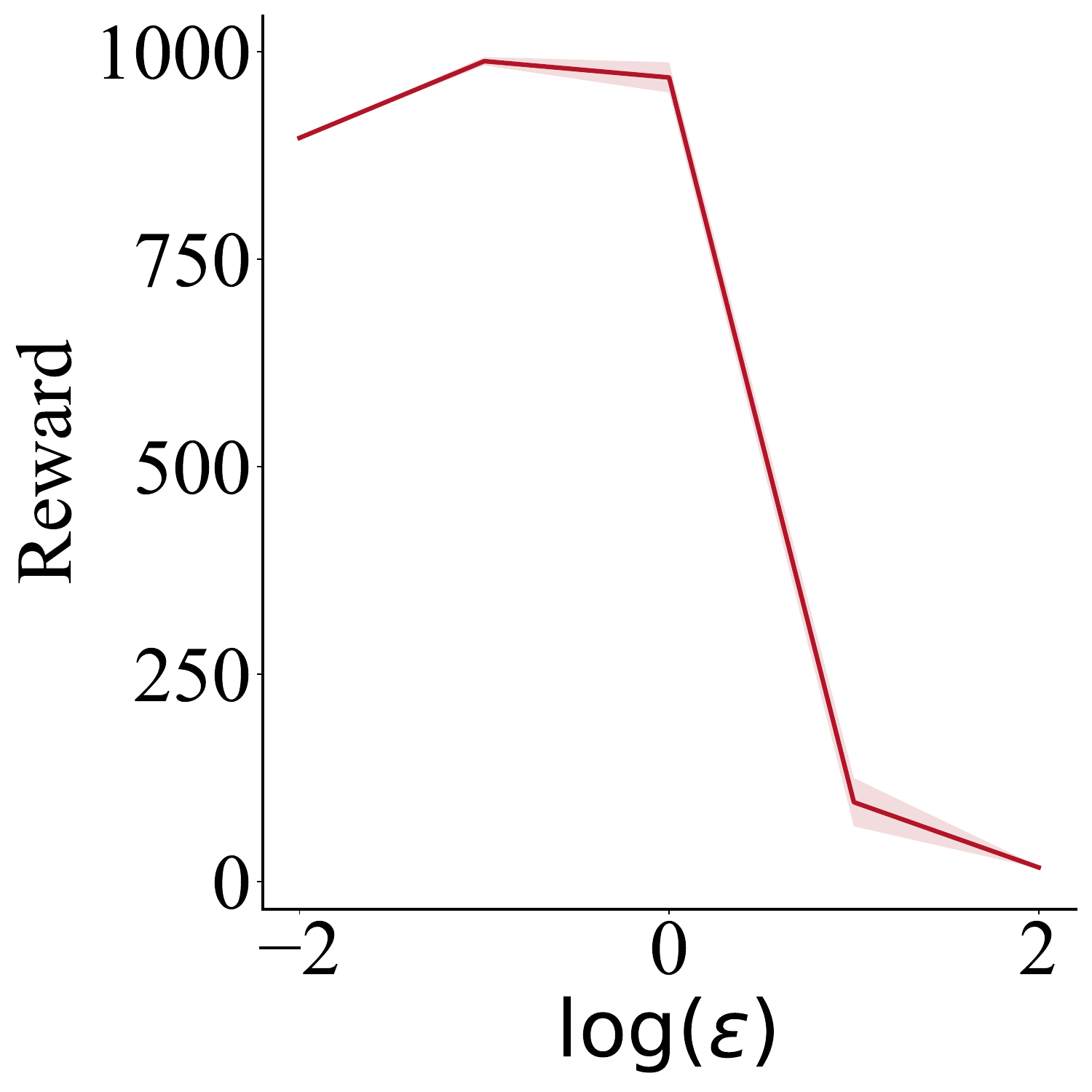}\label{fig:eps}}
    \subfigure[]{\includegraphics[width=.19\textwidth]{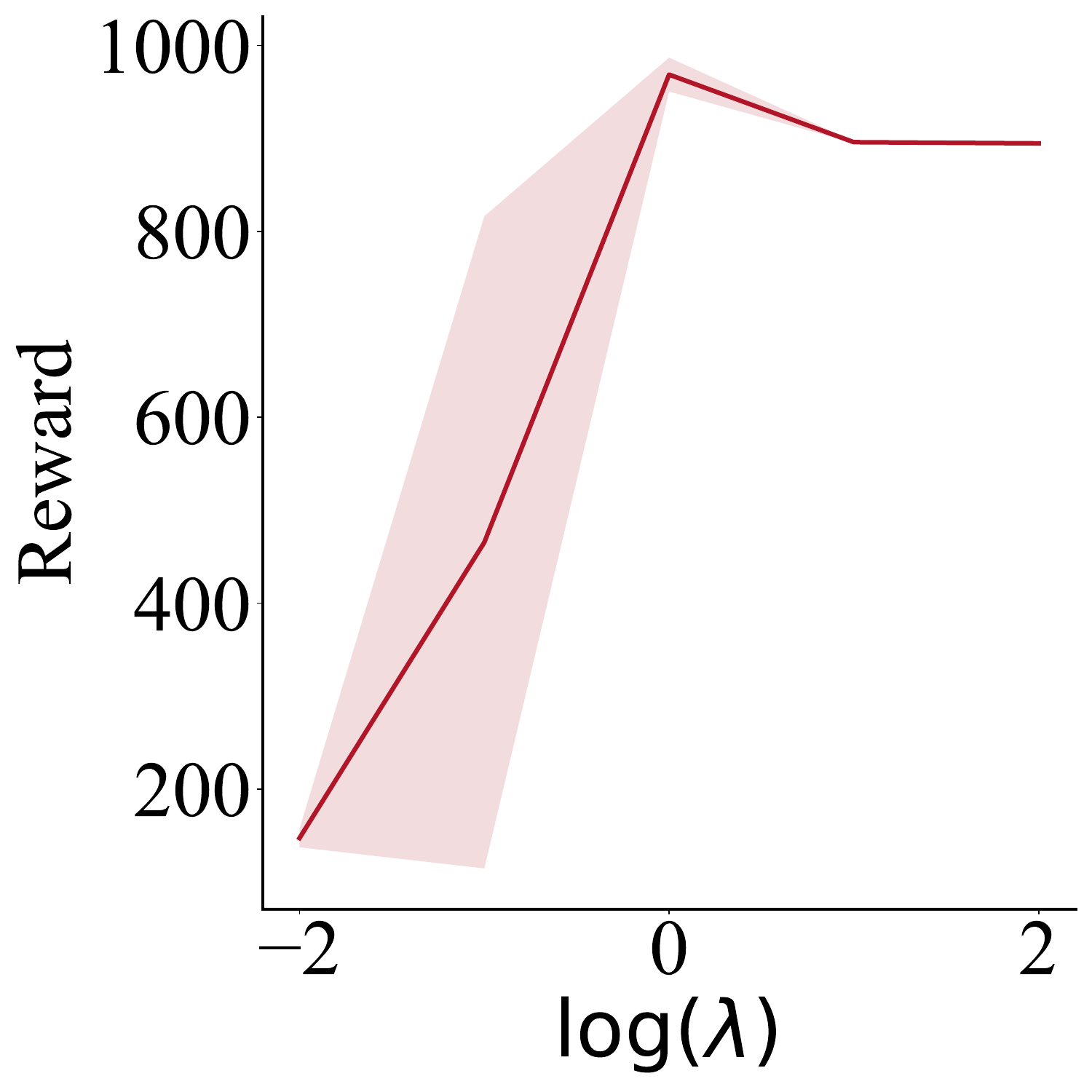}\label{fig:lambda}}
    \subfigure[]{\includegraphics[width=.19\textwidth]{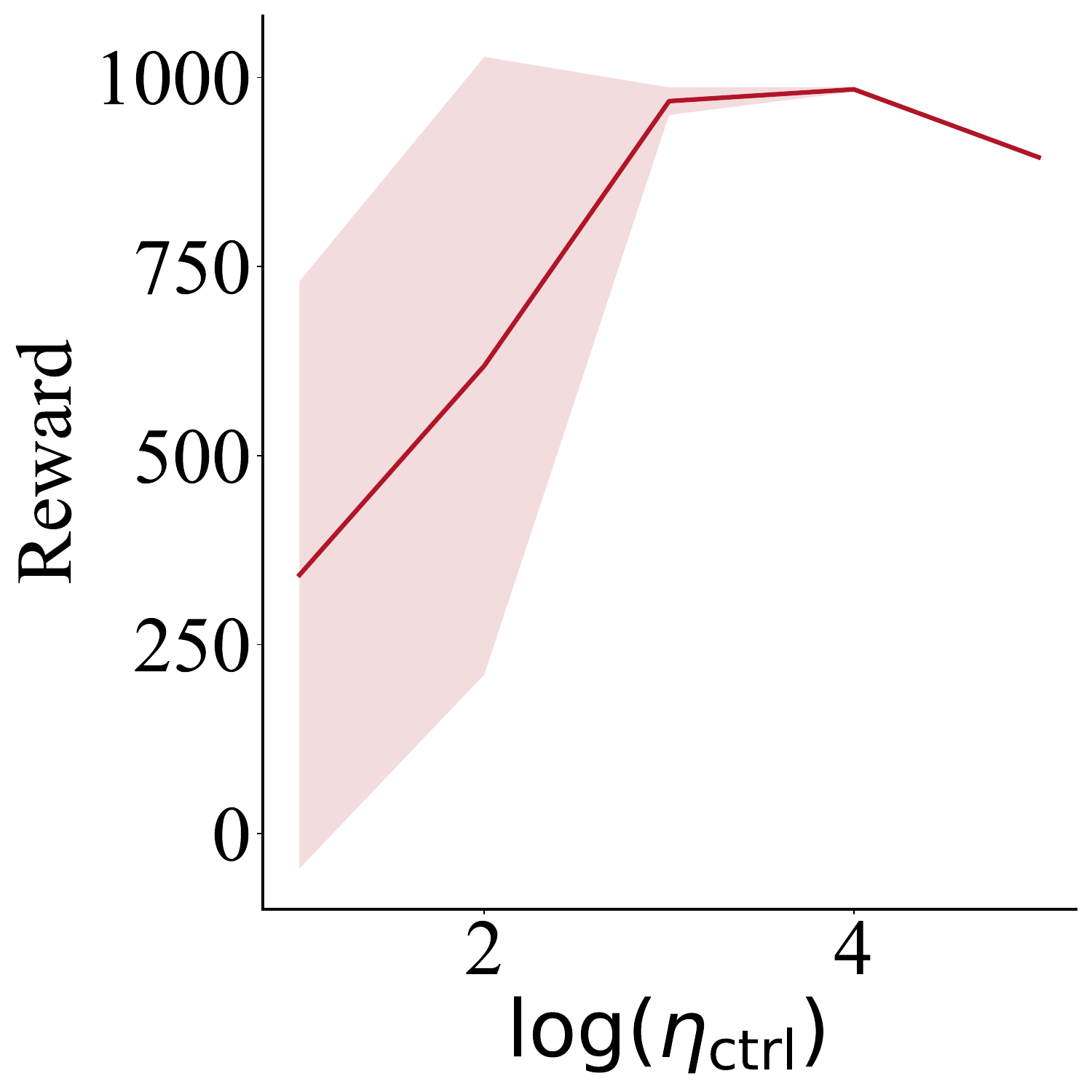}\label{fig:ctrl-weight}}
    \includegraphics[width=.8\textwidth]{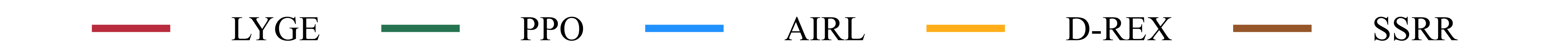}
    \caption{Ablation studies. (a) The converged reward w.r.t. demonstration rewards. (b) The number of samples used before \algo\ converges w.r.t.  demonstration rewards. (c) The converged reward w.r.t. $\epsilon$. (d) The converged reward w.r.t. $\lambda$. (e) The converged reward w.r.t. $\eta_\mathrm{ctrl}$. } 
    \label{fig:ablations}
\end{figure*}

\subsection{Ablation Studies}\label{sec:ablation}

We first show the influence of the optimality of the demonstrations. We use the Inverted Pendulum environment and collect demonstrations with different levels of optimality by varying the distance between the actual goal state and the target state of the demonstrator controller. For each optimality level, we train each algorithm $3$ times with different random seeds and test each converged controller $20$ times. We omit the experiments on CLF-sparse and CLF-dense since their performances are not related to the demonstrations. The results are shown in \Cref{fig:vary-opt}. We observe that \algo\ outperforms other algorithms with demonstrations at different levels of optimality. In addition, we do not observe a significant performance drop of \algo\ as the demonstrations become worse. This is because the quality of the demonstrations only influences the convergence speed of \algo, instead of the controller. PPO's behavior is also consistent since the reward function remains unchanged, but it consistently performs worse than \algo. IL algorithms, however, have a significant performance drop as the demonstrations get worse because they all depend on the quality of the demonstrations.

The quality of the demonstrations can also influence the convergence speed of \algo. In the Inv Pendulum environment, we define the algorithm converges when the reward is larger than $1980$. We plot the number of samples used for convergence w.r.t. the reward of demonstrations in \Cref{fig:conv-time}. It is shown that the better the given demonstrations, the fewer samples \algo\ needs for convergence. 

We also do ablations to investigate the influence of the hyperparameter $\epsilon$. We test \algo\ in the Cart Pole environment and change $\epsilon$ from $0.01$ to $100$. The results are shown in \Cref{fig:eps}, which demonstrates that \algo\ works well when $\epsilon$ is large enough to satisfy the condition introduced in \Cref{thm:clf-converge} in \Cref{sec:app-analysis}, and small enough that it does not make the training very hard. 

The $\lambda$ in \Cref{eq:loss-clf} is another hyperparameter that controls the convergence rate of the learned policy. We test \algo\ in the Cart Pole environment and change $\lambda$ from $0.01$ to $100$. The results are shown in \Cref{fig:lambda}, demonstrating that \algo\ works well with $\lambda$ in some range. If $\lambda$ is too small, the convergence rate is too small and the system cannot be stabilized within the simulation time steps. 
If $\lambda$ is too large, loss \eqref{eq:loss-clf} becomes too hard to converge, so the controller cannot stabilize the system. 

During training, $\eta_\mathrm{ctrl}$ controls the expansion of the trusted tunnel. We do ablations for $\eta_\mathrm{ctrl}$ in the Cart Pole environment and change $\eta_\mathrm{ctrl}$ from $10$ to $10^5$. The results are shown in \Cref{fig:ctrl-weight}, which suggests that \algo\ can work well when the value of $\eta_\mathrm{ctrl}$ is within a certain range. If $\eta_\mathrm{ctrl}$ is too small, the system leaves the trusted tunnel too early instead of smoothly expanding the trusted tunnel. If $\eta_\mathrm{ctrl}$ is too large, exploration is strongly discouraged and the trusted tunnel expands too slowly.


\begin{wrapfigure}{r}{0.49\textwidth}
    \centering
    \subfigure{\includegraphics[width=.2\textwidth]{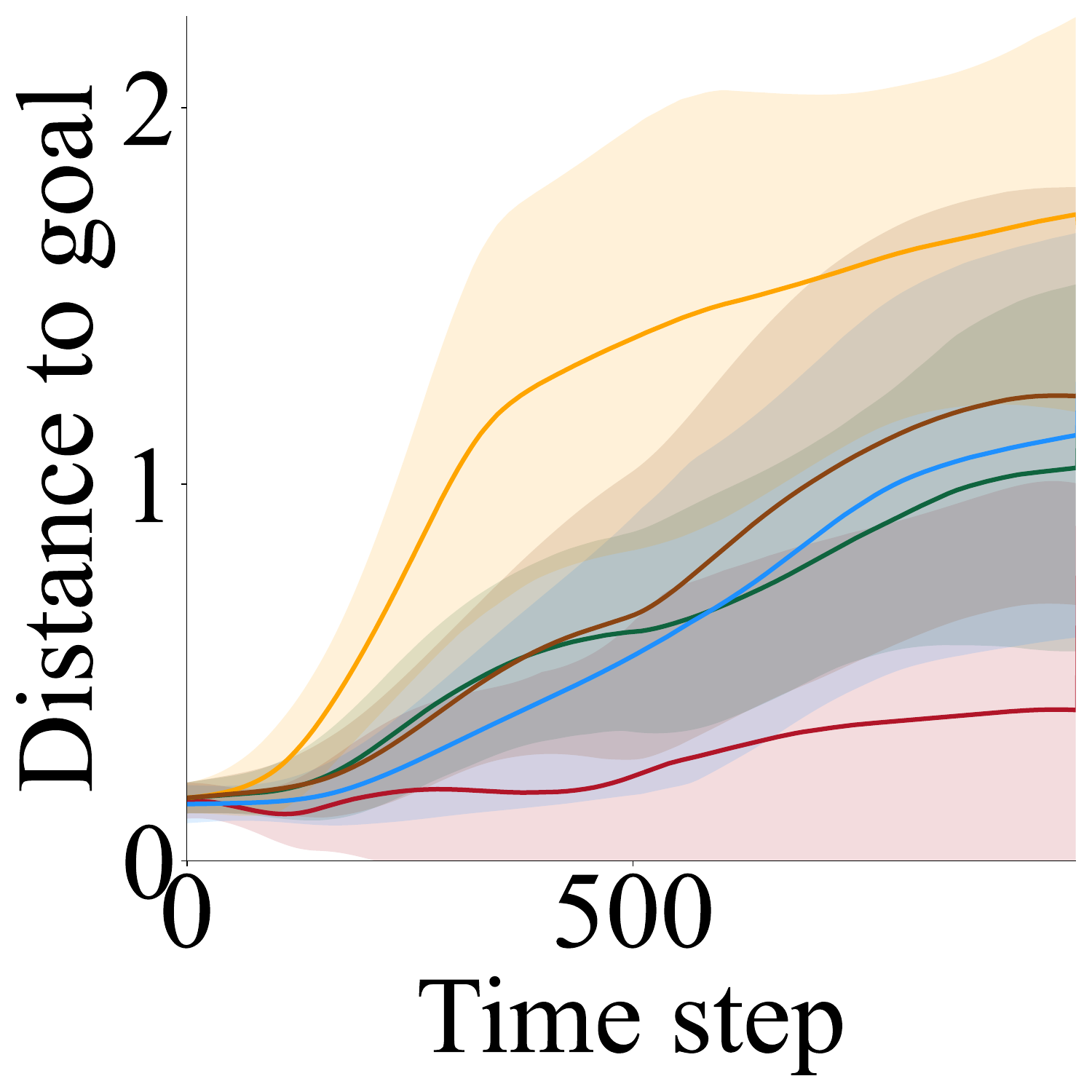}}
    \subfigure{\includegraphics[width=.15\textwidth]{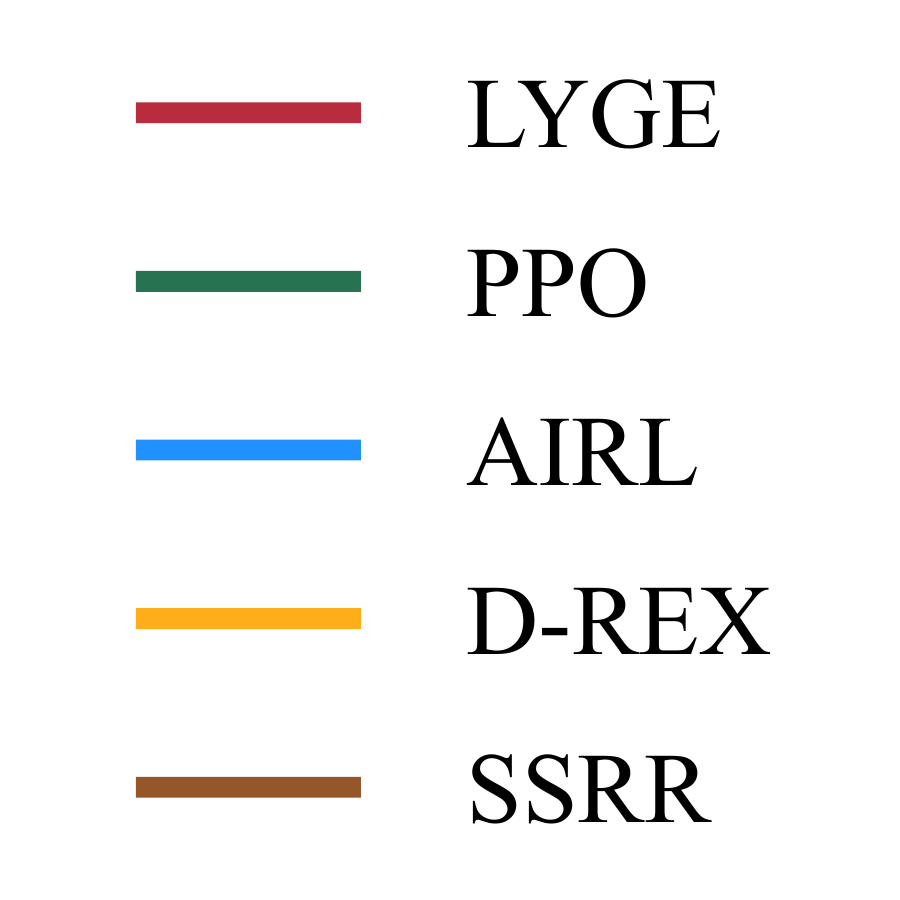}}
    \caption{The tracking error w.r.t. time step in Dubins car path tracking environment.}
    \label{fig:dubins-car}
\end{wrapfigure}

\section{Extensions}\label{sec:extensions}

Our framework is general since it can be directly applied to learn controllers guided by other certificates in environments with unknown dynamics. For example, Control Contraction Metrics (CCMs) are differential analogs of Lyapunov functions (proving stability in the tangent state space). A metric is called CCM if it satisfies a list of conditions, and a valid CCM can guarantee the convergence of tracking controllers. The similarity between CCM and CLF suggests that tracking controllers can also be learned with a similar framework. We change the learning CLF part to learning CCM algorithms~\citep{sun2021learning,chou2021model}, and use the same framework to learn the local dynamics, a tracking controller, and a CCM to guide the exploration of the tracking controller. We test this modification in a Dubins car path tracking environment. As shown in \Cref{fig:dubins-car}, our algorithm outperforms the baselines. We explain the details in \Cref{sec:extension-detail}.


\section{Conclusion}\label{sec:conclusion}

We propose a general learning framework, \algo, for learning stabilizing controllers in high-dimensional environments with unknown dynamics. \algo\ iteratively fits the local dynamics to form a trusted tunnel, and learns a control policy with a CLF to guide the exploration and expand the trusted tunnel toward the goal. Upon convergence, the controller stabilizes the closed-loop system at the goal point. We provide experimental results to demonstrate that \algo\ performs comparably or better than the baseline RL, IL, and neural certificate methods. We also demonstrate that the same framework can be applied to learn other certificates in environments with unknown dynamics. 

Our framework has a few \textbf{limitations}: we require a set of demonstrations for initialization in high-dimensional systems, although they can be potentially imperfect. Without them, \algo\ may take a long time to expand the trusted tunnel to the goal. In addition, we need Lipschitz assumptions for the dynamics to derive the theoretical results. If the dynamics do not satisfy the Lipschitz assumptions, the learned CLF might be invalid even inside the trusted tunnel. 
Moreover, although we observe the convergence of the loss terms in training on all our case studies, it is hard to guarantee that the loss always converges on any system.
Finally, if we desire a fully validated Lyapunov function, we need to employ formal verification tools, and we provide a detailed discussion in \Cref{sec:app-discussions}.

\acks{This work was partly supported by the National Science Foundation (NSF) CAREER Award \#CCF-2238030 and the MIT-DSTA program. Any opinions, findings, and conclusions or recommendations expressed in this publication are those of the authors and don’t necessarily reflect the views of the sponsors.}

\bibliography{main}

\newpage
\appendix

\section{Control Lyapunov Functions}\label{sec:app-clf}

In this section, we review the stability results using control Lyapunov functions. We provide the definition of CLF and the stability results in \Cref{sec:preliminary}. Here we provide the results formally. 

\begin{proposition}\label{thm:clf-stable}
    Given a set $\mathcal{G}\subset\mathcal{X}$ such that $\xg\in\mathcal{G}$. Suppose there exists a CLF $V$ on $\mathcal{G}\subseteq\mathcal{X}$ with a constant $\lambda\in (0, 1)$. If $\mathcal{G}$ is forward invariant\footnote{A set $\mathcal{G}$ is forward invariant for \eqref{eq:dynamics} if $x(0)\in\mathcal{G}\implies x(t)\in\mathcal{G}$ for all $t>0$.}, then $\xg$ is asymptotically stable for the closed-loop system under $u\in\mathcal{K}(x)=\left\{u\; |\; V\left(h(x,u)\right)\leq\lambda V(x)\right\}$ starting from initial set $\mathcal{X}_0\subseteq\mathcal{G}$.
\end{proposition}

\begin{proof}
    The proof follows from Definition 2.18 and Theorem 2.19 in \cite{grune2017nonlinear}. Condition~\eqref{eq:clf-pos-cond} is the same as condition (i) in \cite{grune2017nonlinear}, and for condition 2, let $u^*(x)=\inf_u V(h(x,u))$, from Condition~\eqref{eq:clf-decrease-cond}, we have:
    \begin{equation}
        V(h(x,u^*(x)))\leq\lambda V(x),\quad \forall x\in\mathcal{G}.
    \end{equation}
    Therefore,
    \begin{equation}
        \begin{aligned}
            V(h(x,u^*(x)))\leq V(x)-(1-\lambda)V(x)\leq V(x)-(1-\lambda)\underline\alpha(\|x-\xg\|),\quad \forall x\in\mathcal{G},
        \end{aligned}
    \end{equation}
    where the second equation follows condition~\eqref{eq:clf-pos-cond}. Define $g(x)=(1-\lambda)x$. Since $\lambda\in(0,1)$, it follows that $g(x)\in\mathcal{K}_\infty$. Let $\alpha_V(x-\xg)=g\circ\underline{\alpha}(\|x-\xg\|)$, we have $\alpha_V\in\mathcal{K}_\infty$ and
    \begin{equation}
        V(h(x,u^*(x)))\leq V(x)-\alpha_V(\|x-\xg\|),\quad\forall x\in\mathcal{G},
    \end{equation}
    which aligns with condition (ii) in \cite{grune2017nonlinear}. Note that $\mathcal{G}$ is assumed forward invariant. Therefore, following Theorem 2.19 in \cite{grune2017nonlinear}, the system is asymptotically stable starting from $\mathcal{X}_0\in\mathcal{G}$.
\end{proof}



\section{Analysis of \algo}\label{sec:app-analysis}

We first provide the workflow of \algo\ in \Cref{alg:lyge}.


\begin{algorithm}[h]
   \caption{\algo}
   \label{alg:lyge}
\begin{algorithmic}
   \STATE {\bfseries Input:} Demonstrations $\mathcal{D}^0$, set of initial states $\mathcal{X}_0$
   \STATE Learn initial controller $\pi_\mathrm{init}$ using IL
   \STATE Initialize $\hat{h}^0_\psi$ with an NN
   \STATE Initialize $V^0_\theta$ with the structure provided in \Cref{sec:algo}
   \STATE Initialize the controller with $\pi^0_\phi=\pi_\mathrm{init}$
   \FOR{$\tau$ in $0,1,...$}
   \STATE Sample states from $\mathcal{D}^\tau_x$
   \STATE Update $\hat{h}^\tau_\psi$ using MSE loss
   \STATE Update $V^\tau_\theta$ and $\pi^\tau_\phi$ using \eqref{eq:loss-clf} and \eqref{eq:loss-ctrl}
   \STATE Collect more trajectories $\Delta D^\tau$ by starting from $\mathcal{D}^\tau_x$ and following $\pi^\tau_\phi$
   \STATE Construct $\mathcal{D}^{\tau+1}=\mathcal{D}^\tau\cup\Delta D^\tau$ 
   \ENDFOR
\end{algorithmic}
\end{algorithm}

Next, we provide several results to discuss the efficacy of \algo. We start by reviewing some definitions and results on Lipschitz continuous and Lipschitz smoothness. 

\begin{definition}
A function $f: \mathbb{R}^n\rightarrow\mathbb{R}^m$ is Lipschitz-continuous with constant $L$ if 
\begin{equation}
    ||f(x)-f(y)||\leq L\|x-y\|, \quad \forall x,y\in \mathbb{R}^n.
\end{equation}
\end{definition}

\begin{definition}
A continuously differentiable function $f: \mathbb{R}^n\rightarrow\mathbb{R}^m$ is Lipschitz-smooth with constant $L$ if 
\begin{equation}
    ||\nabla f(x)-\nabla f(y)||\leq L\|x-y\|, \quad \forall x,y\in \mathbb{R}^n.
\end{equation}
\end{definition}

\begin{lemma}\label{lem:lipschitz-cauchy}
If a continuously differentiable function $f:\mathbb{R}^n\rightarrow\mathbb{R}^m$ is Lipschitz-smooth with constant $L$, then the following inequality holds for all $x,y\in\mathbb{R}^n$:
\begin{equation}
    \left(\nabla f(x)-\nabla f(y)\right)^T(x-y)\leq L\|x-y\|^2.
\end{equation}
\end{lemma}

\begin{proof}
The proof is straightforward that
\begin{equation}
    \begin{aligned}
    \left(\nabla f(x)-\nabla f(y)\right)^T(x-y)\leq \|\nabla f(x)-\nabla f(y)\|\cdot \|x-y\|\leq L\|x-y\|^2,
    \end{aligned}
\end{equation}
where the first equation follows the Cauchy-Schwarz inequality, and the second inequality comes from the definition of Lipschitz-smoothness. 
\end{proof}

\begin{lemma}\label{lem:lipschitz-inequality}
If function $f:\mathbb{R}^n\rightarrow\mathbb{R}^m$ is Lipschitz-smooth with constant $L$, then the following inequality holds for all $x,y\in\mathbb{R}^n$: 
\begin{equation}
    f(y)\leq f(x)+\nabla f(x)^T(y-x)+\frac{L}{2}\|y-x\|^2.
\end{equation}
\end{lemma}

\begin{proof}
Define $g(t)=f(x+t(y-x))$. If $f$ is Lipschitz-smooth with constant $L$, then from Lemma~\ref{lem:lipschitz-cauchy}, we have
\begin{equation}
    \begin{aligned}
    &g'(t)-g'(0)\\
    &=\left(\nabla f(x+t(y-x))-\nabla f(x)\right)^T(y-x)\\
    &=\frac{1}{t} \left(\nabla f(x+t(y-x))-\nabla f(x)\right)^T((x+t(y-x))-x)\\
    &\leq \frac{L}{t}\|t(y-x)\|^2=tL\|y-x\|^2.
    \end{aligned}
\end{equation}
We then integrate this equation from $t=0$ to $t=1$:
\begin{equation}
    \begin{aligned}
    f(y)&=g(1)=g(0)+\int_0^1 g'(t)dt\\
    &\leq g(0)+\int_0^1g'(0)dt+\int_0^1tL\|y-x\|^2dt\\
    &=g(0)+g'(0)+\frac{L}{2}\|y-x\|^2\\
    &=f(x)+\nabla f(x)^T(y-x)+\frac{L}{2}\|y-x\|^2.
    \end{aligned}
\end{equation}
\end{proof}

Now, we provide the following theorem to show the convergence of \algo. We say that $\mathcal{L}^\tau_\mathrm{CLF}$ $\epsilon'$-\textit{robustly converges} if (1) $V^\tau_\theta(\xg)<\nu$; (2) $V^\tau_\theta(x)>\nu$ for all $x\in\mathcal{X}\setminus\xg$; (3) $\epsilon+V^\tau_\theta(\hat{h}^\tau_\psi(x,\pi^\tau_\phi(x)))-\lambda V^\tau_\theta(x)\leq\epsilon'$ for all $x\in \mathcal D^\tau_x$. Here, $\epsilon'>0$ is an arbitrarily small number. 

\begin{theorem}\label{thm:clf-converge}
    Let $L_\pi$, $L_V$ be the Lipschitz constants of the learned controller $\pi^\tau_\phi$ and the gradient of the learned CLF $V^\tau_\theta$, respectively. Furthermore, let $\sigma\geq\left\|\nabla V_\theta^\tau\right\|$ be the upper bound of the gradient of $V_\theta^\tau$. Choose $\epsilon\geq \omega\sigma+\frac{L_V}{2}\omega^2+\gamma\lambda(\sigma+\frac{L_V}{2}\gamma)+(1+L_\pi)\gamma L_h(\sigma+\frac{L_V}{2}(1+L_\pi)\gamma L_h)+\epsilon'$. If $\mathcal{L}^\tau_\mathrm{CLF}$ $\epsilon'$-robustly converges in each iteration, then \algo\ converges and returns a stabilizing controller $\pi^*$ that can be trusted within the converged trusted tunnel $\mathcal{H}^*$, where $\mathcal{H}^*$ contains all closed-loop trajectories starting from $\mathcal{X}_0$.
\end{theorem}

\Cref{thm:clf-converge} shows that with smooth dynamics and smooth NNs, if the margin $\epsilon$ is chosen to be large enough and the training loss $\mathcal{L}^\tau_\mathrm{CLF}$ is small, we can conclude that the algorithm converges and the system is asymptotically stable at $\xg$. In our implementation, we use spectral normalization to limit the Lipschitz constants of the NNs. We also increase the amount of data collected in the exploration phase and use large NNs to decrease $\gamma$ and $\omega$. In this way, we can make $\epsilon$ a reasonably small value. 

Now, we provide the proof of \Cref{thm:clf-converge}. We start by introducing several lemmas. First, we show that the learned CLF satisfies the CLF conditions~\eqref{eq:clf-cond} within the trusted tunnel $\mathcal{H}^\tau$.


\begin{lemma}\label{thm:correct-clf}
    Under the assumptions of \Cref{thm:clf-converge}, we have in iteration $\tau$, 
    \begin{equation}
        V^\tau_\theta(x+h(x,\pi^\tau_\phi(x)))\leq\lambda V^\tau_\theta(x),\forall x\in\mathcal{H}^\tau.
    \end{equation}
\end{lemma}
\begin{proof}
    For arbitrary $x\in\mathcal{H}^\tau$, let $\bar x\in\mathcal{D}^\tau_x$ be the closest point to $x$ in the dataset $\mathcal{D}^\tau_x$, i.e., $\bar x = \textrm{arg}\min_{y\in \mathcal D^\tau_x}\|y- x\|$. Using Lemma~\ref{lem:lipschitz-inequality} and the definition of the trusted tunnel, we have:
    \begin{equation}\label{eq:v-error}
        \begin{aligned}
            V^\tau_\theta(\bar x)-V^\tau_\theta(x)&=V^\tau_\theta(x+\bar x-x)-V^\tau_\theta(x)\leq (\bar x-x)^\top\nabla V^\tau_\theta(x)+\frac{L_V}{2}\|\bar x-x\|^2\\
            &\leq \|\bar x-x\|\|\nabla V^\tau_\theta(x)\|+\frac{L_V}{2}\|\bar x-x\|^2\leq\gamma\sigma+\frac{L_V}{2}\gamma^2.
        \end{aligned}
    \end{equation}
    Using the Lipschitz continuity of the dynamics $h$ and the controller $\pi^\tau_\phi$, we have
    \begin{equation}\label{eq:h-error}
        \begin{aligned}
            \|h(x,\pi^\tau_\phi(x))-h(\bar x,\pi^\tau_\phi(\bar x)\|&\leq L_h\|(x,\pi^\tau_\phi(x))-(\bar x,\pi^\tau_\phi(\bar x))\|=L_h\|(x-\bar x,\pi^\tau_\phi(x)-\pi^\tau_\phi(\bar x))\| \\
            &\leq L_h(\|x-\bar x\|+\|\pi^\tau_\phi(x)-\pi^\tau_\phi(\bar x)\|)\leq(1+L_\pi)\gamma L_h.
        \end{aligned}
    \end{equation}
    Using Lemma~\ref{lem:lipschitz-inequality} and the bounded gradient of $V$, and applying \eqref{eq:v-error} and \eqref{eq:h-error}, we obtain that for any $x\in\mathcal{H}^\tau$,
    \begin{equation}
        \begin{aligned}
            &V^\tau_\theta\left(h(x,\pi^\tau_\phi(x))\right)-\lambda V^\tau_\theta(x) \\
            =&V^\tau_\theta\left(h(x,\pi^\tau_\phi(x))-h(\bar x,\pi^\tau_\phi(\bar x))+h(\bar x,\pi^\tau_\phi(\bar x))\right)-\lambda V^\tau_\theta(x) \\
            \leq& V^\tau_\theta\left(h(\bar x,\pi^\tau_\phi(\bar x))\right)+\|h(x,\pi^\tau_\phi(x))-h(\bar x,\pi^\tau_\phi(\bar x))\|\sigma \\
            &+\frac{L_V}{2}\|h(x,\pi^\tau_\phi(x))-h(\bar x,\pi^\tau_\phi(\bar x))\|^2-\lambda V^\tau_\theta(\bar x)+\lambda(V^\tau_\theta(\bar x)-V^\tau_\theta(x))\\
            \leq&V^\tau_\theta\left(h(\bar x,\pi^\tau_\phi(\bar x))\right)-\lambda V^\tau_\theta(\bar x)+\gamma\lambda(\sigma+\frac{L_V}{2}\gamma)+(1+L_\pi)\gamma L_h(\sigma+\frac{L_V}{2}(1+L_\pi)\gamma L_h).
        \end{aligned}
    \end{equation}
    Then, we take the error of the learned dynamics into consideration. Using the error bound of the learned dynamics, we have:
    \begin{equation}
        \begin{aligned}
            &V^\tau_\theta\left(\hat{h}^\tau_\psi(\bar x,\pi^\tau_\phi(\bar x))\right)-\lambda V^\tau_\theta(\bar x)\\
            =&V^\tau_\theta\left(h^\tau_\psi(\bar x,\pi^\tau_\phi(\bar x))-h^\tau_\psi(\bar x,\pi^\tau_\phi(\bar x))+\hat{h}^\tau_\psi(\bar x,\pi^\tau_\phi(\bar x))\right)-\lambda V^\tau_\theta(\bar x)\\
            \geq&V^\tau_\theta\left(h^\tau_\psi(\bar x,\pi^\tau_\phi(\bar x)))\right)-\|\hat{h}^\tau_\psi(\bar x,\pi^\tau_\phi(\bar x))-h^\tau_\psi(\bar x,\pi^\tau_\phi(\bar x))\|\sigma \\
            &-\frac{L_V}{2}\|\hat{h}^\tau_\psi(\bar x,\pi^\tau_\phi(\bar x))-h^\tau_\psi(\bar x,\pi^\tau_\phi(\bar x))\|^2-\lambda V^\tau_\theta(\bar x)\\
            \geq&V^\tau_\theta\left(h^\tau_\psi(\bar x,\pi^\tau_\phi(\bar x)))\right)-\lambda V^\tau_\theta(\bar x)-\omega\sigma-\frac{L_V}
            {2}\omega^2.
        \end{aligned}
    \end{equation}
    Using the assumption that
    \begin{equation}
        V^\tau_\theta(\hat{h}^\tau_\psi(\bar x,\pi^\tau_\phi(\bar x)))-\lambda V^\tau_\theta(\bar x)+\epsilon\leq\epsilon',\quad \forall \bar x\in \mathcal D^\tau_x,
    \end{equation}
    we have
    \begin{equation}
        \begin{aligned}
            &V^\tau_\theta\left(h^\tau_\psi(\bar x,\pi^\tau_\phi(\bar x)))\right)-\lambda V^\tau_\theta(\bar x)\\
            \leq&V^\tau_\theta\left(\hat{h}^\tau_\psi(\bar x,\pi^\tau_\phi(\bar x))\right)-\lambda V^\tau_\theta(\bar x)+\omega\sigma+\frac{L_V}{2}\omega^2\\
            \leq&\epsilon'-\epsilon+\omega\sigma+\frac{L_V}{2}\omega^2.\\
        \end{aligned}
    \end{equation}
    Therefore,
    \begin{equation}
        \begin{aligned}
            &V^\tau_\theta\left(x+h(x,\pi^\tau_\phi(x))\right)-\lambda V^\tau_\theta(x)\\
            \leq&\epsilon'-\epsilon+\omega\sigma+\frac{L_V}{2}\omega^2+\gamma\lambda(\sigma+\frac{L_V}{2}\gamma)+(1+L_\pi)\gamma L_h(\sigma+\frac{L_V}{2}(1+L_\pi)\gamma L_h)\\
            \leq& 0.
        \end{aligned}
    \end{equation}
\end{proof}

Lemma~\ref{thm:correct-clf} suggests that under the assumptions of \Cref{thm:clf-converge}, the learned CLF satisfies the CLF condition \eqref{eq:clf-decrease-cond} in the trusted tunnel $\mathcal{H}^\tau$.


Next, we discuss the \emph{growth} of the trusted tunnel $\mathcal{H}^\tau$ in each iteration.

\begin{lemma}\label{thm:h-grow}
    Let $T$ be the simulation horizon, and $x^\tau\in\mathcal{H}^\tau$ be the state that has the minimum value of the CLF at iteration $\tau$, \ie,
    \begin{equation}
        x^\tau=\arg\min_{x\in\mathcal{H}^\tau} V^\tau_\theta(x).
    \end{equation}
    If $\xg\not\in\mathcal{H}^\tau$ and at least one of the sampled state $x_d^\tau$ satisfies
    \begin{equation}\label{eq:sample-xd}
        \|x_d^\tau-x^\tau\|<\frac{-\sigma+\sqrt{\sigma^2+2L_V\left(\frac{1}{\lambda^T}-1\right)V^\tau_\theta(x^\tau)}}{L_V},
    \end{equation}
    then during the exploration period, we have $\mathcal{H}^\tau\subsetneqq\mathcal{H}^{\tau+1}$. 
\end{lemma}

\begin{proof} 
Using the definition of the trusted tunnel $\mathcal{H}^\tau$ and the fact that $\mathcal{D}^{\tau}\subset\mathcal{D}^{\tau+1}$, we have $\mathcal{H}^\tau\subseteq\mathcal{H}^{\tau+1}$. Let $x^\tau\in\mathcal{H}^\tau$ be the state that has the minimum value of the CLF, \ie,
\begin{equation}
    x^\tau=\arg\min_{x\in\mathcal{H}^\tau} V^\tau_\theta(x).
\end{equation}
During the exploration process, let $x^\tau_d$ be a sampled initial state. Then, it follows from Lemma~\ref{thm:correct-clf} that
\begin{equation}\label{eq:v_x_d-bound}
    \begin{aligned}
        V^\tau_\theta(x_d^\tau)&\leq V^\tau_\theta(x^\tau)+(x_d^\tau-x^\tau)^\top\nabla V^\tau_\theta(x^\tau)+\frac{L_V}{2}\|x_d^\tau-x^\tau\|^2,
    \end{aligned}
\end{equation}
If the trajectory leaves the trusted tunnel $\mathcal{H}^\tau$, then the claim is true. Otherwise, the trajectory stays in $\mathcal{H}^\tau$ in the simulation horizon $T$. Then, using Lemma~\ref{thm:correct-clf}, we have
\begin{equation}\label{eq:v-trajectory}
    V^\tau_\theta(x_d^\tau(T))\leq\lambda V^\tau_\theta(x_d^\tau(T-1))\leq\cdots\leq\lambda^TV^\tau_\theta(x_d^\tau(0))=\lambda^TV^\tau_\theta(x_d^\tau).
\end{equation}
Additionally, since the trajectory stays in $\mathcal{H}^\tau$, we have 
\begin{equation}\label{eq:v-T}
    V^\tau_\theta(x_d^\tau(T))\geq V^\tau_\theta(x^\tau).
\end{equation}
Using \eqref{eq:v_x_d-bound}, \eqref{eq:v-trajectory}, and \eqref{eq:v-T}, we obtain
\begin{equation}
    \lambda^T\left(V^\tau_\theta(x^\tau)+(x_d^\tau-x^\tau)^\top\nabla V^\tau_\theta(x^\tau)+\frac{L_V}{2}\|x_d^\tau-x^\tau\|^2\right)\geq V^\tau_\theta(x^\tau).
\end{equation}
Note that $V^\tau_\theta$ has bounded gradients $\sigma$. Therefore,
\begin{equation}
    \frac{L_V}{2}{\|x_d^\tau-x^\tau\|}^2+\sigma \|x_d^\tau-x^\tau\|-\left(\frac{1}{\lambda^T}-1\right)V^\tau_\theta(x^\tau)\geq0,
\end{equation}
which implies
\begin{equation}
    \|x_d^\tau-x^\tau\|\geq\frac{-\sigma+\sqrt{\sigma^2+2L_V\left(\frac{1}{\lambda^T}-1\right)V^\tau_\theta(x^\tau)}}{L_V}.
\end{equation}
This is the necessary condition for the trajectory to stay in $\mathcal{H}$. Otherwise, we have $V^\tau_\theta(x^\tau_d(T))<V^\tau_\theta(x^\tau)$, which violates the assumption that $x^\tau$ is the minima of $V^\tau_\theta$ in $\mathcal{H}^\tau$. Therefore, if we sample an initial state $x_d^\tau$ with
\begin{equation}
    \|x_d^\tau-x^\tau\|<\frac{-\sigma+\sqrt{\sigma^2+2L_V\left(\frac{1}{\lambda^T}-1\right)V^\tau_\theta(x^\tau)}}{L_V},
\end{equation}
the trajectory will leave $\mathcal{H}^\tau$, which implies $\mathcal{H}^\tau\subsetneqq\mathcal{H}^{\tau+1}$.

\end{proof}

Note that the RHS of inequality \eqref{eq:sample-xd} grows to infinity as $T$ grows. Lemma~\ref{thm:h-grow} shows that the trusted tunnel $\mathcal{H}^\tau$ continues to grow when $\xg$ is not inside. Also, it shows that the trusted tunnel $\mathcal{H}^\tau$ cannot converge to some $\mathcal{H}^*$ such that $\xg\not\in\mathcal{H}^*$. 

Now, we are ready to provide the proof of \Cref{thm:clf-converge}.


\begin{proof}
    First, using Lemma~\ref{thm:h-grow}, we know that the size of $\mathcal{H}^\tau$ increases monotonically. Since the size of $\mathcal{H}^\tau$ is upper-bounded by the compact state space, using Monotone Convergence Theorem, we know that $\mathcal{H}^\tau$ will converge to some set $\mathcal{H}^*$. Then, using Lemma~\ref{thm:correct-clf}, we know that the CLF conditions~\eqref{eq:clf-cond} are satisfied in $\mathcal{H}^*$\footnote{Note that although $V^\tau_\theta(\xg)$ is not exactly zero, $\xg$ is still the global minimum of $V^\tau_\theta$, and therefore, the closed-loop system will still converge to $\xg$.}. Using Lemma~\ref{thm:h-grow}, we know that $\xg\in\mathcal{H}^*$. In addition, as $\mathcal{H}^\tau$ converges to $\mathcal{H}^*$, we know that starting from any initial states in $\mathcal{H}^*$, the agent cannot leave $\mathcal{H}^*$. Therefore, using Proposition~\ref{thm:clf-stable}, the system is asymptotically stable at $\xg$.
\end{proof}

\section{Experiments}\label{sec:app-experiments}

In this section, we provide additional experimental details and results. We provide the code of our experiments in the file `lyge.zip' in the supplementary materials. 

\subsection{Experimental Details}

Here we introduce the details of the experiments, including the implementation details of \algo\ and the baselines, choice of hyper-parameters, and detailed introductions of environments. The experiments are run on a 64-core AMD 3990X CPU @ $3.60\mathrm{GHz}$ and four NVIDIA RTX A4000 GPUs (one GPU for each training job).

\subsubsection{Implementation details}\label{sec:app-implementation}
\paragraph{Implementation of \algo} 
Our framework contains three models: the dynamics model $\hat h^\tau_{\psi}(x,u)$, the CLF $V^\tau_\theta(x)=x^\top S^\top S x+p_\mathrm{NN}(x)^\top p_\mathrm{NN}(x)$, and the controller $\pi^\tau_\phi(x)$. $\hat h^\tau_{\psi}(x,u)$, $p_\mathrm{NN}(x)$, and $\pi^\tau_\phi(x)$ are all neural networks with two hidden layers with size $128$ and $\mathrm{Tanh}$ as the activation function. $S\in\mathbb{R}^{n_x\times n_x}$ is a matrix of parameters. To limit the Lipschitz constant of the learned models $L_V$ and $L_\pi$, we add spectral normalization~\citep{miyato2018spectral} to each layer in the neural networks. We implement our algorithm in the PyTorch framework~\citep{paszke2019pytorch} based on the rCLBF repository\footnote{https://github.com/MIT-REALM/neural\_clbf (BSD-3-Clause license)}~\citep{dawson2022safe}. During training, we use ADAM~\citep{kingma2014adam} as the optimizer to optimize the parameters of the neural networks. The loss function used in training the controller and the CLF is 
\begin{equation}
\mathcal{L}^\tau=\eta_\mathrm{goal}\mathcal{L}^\tau_\mathrm{goal}+\eta_\mathrm{pos}\mathcal{L}^\tau_\mathrm{pos}+\eta_\mathrm{ctrl}\mathcal{L}^\tau_\mathrm{ctrl},
\end{equation}
where $\eta_\mathrm{goal}$, $\eta_\mathrm{pos}$, $\eta_\mathrm{ctrl}$ are hyper-parameters, which we will further introduce in Appendix \ref{sec:choice-param}, and
\begin{equation}
    \begin{aligned}
    & \mathcal{L}^\tau_\mathrm{goal}=V^\tau_\theta(\xg)^2,\\
    & \mathcal{L}^\tau_\mathrm{pos}=\frac{1}{N}\sum_{x\in\mathcal{D}^\tau_x}\max\left[\epsilon+V^\tau_\theta\left(\hat{h}^\tau_\psi(x,\pi^\tau_\phi(x))\right)-\lambda V^\tau_\theta(x), 0\right],\\
    & \mathcal{L}^\tau_\mathrm{ctrl}=\frac{1}{N}\sum_{x\in\mathcal{D}^\tau_x}\left\|\pi^\tau_\phi(x)-\pi_\mathrm{init}(x)\right\|^2+\|\pi^\tau_\phi(x)-\pi^{\tau-1}_\phi(x)\|^2.
    \end{aligned}
\end{equation}
Note that there is another term in the loss: $\frac{1}{N}\sum_{x\in\mathcal{X}\setminus\xg}\max\left[\nu-V^\tau_\theta(x),0\right]$. However, this loss term is often $0$ in the training so we omit the discussion of this term here. In our implementation, we add the term $\|\pi^\tau_\phi(x)-\pi_\mathrm{init}(x)\|^2$ to further limit the exploration of $\pi^\tau_\phi$, and also make the training more stable, as there is a non-changing reference for $\pi^\tau_\phi$.

We provide more details about the exploration here. Starting from some states in the the trusted tunnel $x(0)\in\mathcal{H}^\tau$, following the newly updated controller in this iteration $\pi^\tau_\phi$, the trajectories collected are $\{x(0),\pi^\tau_\phi(x(0)),x(1),\pi^\tau_\phi(x(1)),x(2),\pi^\tau_\phi(x(2)),\dots\}$, where $x(t+1)=h(x(t),\pi^\tau_\phi(x(t)))$. The trajectory ends when the maximum simulation time step $T$ is reached or the states are no longer safe. In each iteration, \algo\ sample $8000$ environment steps. 

\paragraph{Implementation of the baselines}
We implement PPO based on the open-source python package stablebaselines3\footnote{https://github.com/DLR-RM/stable-baselines3 (MIT license)}~\citep{stable-baselines3}, AIRL based on the open-source python package Imitation\footnote{https://github.com/HumanCompatibleAI/imitation (MIT license)}~\citep{wang2020imitation}, and D-REX and SSRR based on their official implementations~\footnote{https://github.com/dsbrown1331/CoRL2019-DREX (MIT license)}\footnote{https://github.com/CORE-Robotics-Lab/SSRR}, with some adjustments based on the CAIL repository~\footnote{https://github.com/Stanford-ILIAD/Confidence-Aware-Imitation-Learning (MIT license)}~\citep{zhang2021confidence}. All the neural networks in the baselines, including the actor, the critic, the discriminator, and the reward module, have two hidden layers with size $128$ and $\mathrm{Tanh}$ as the activation function. We use ADAM~\citep{kingma2014adam} as the optimizer with a learning rate $3\times 10^{-4}$ to optimize the parameters of the neural networks. For CLF-sparse and CLF-dense, we use the same NN structures as \algo, but pre-train the dynamics $\hat{h}^\tau_\psi(x,u)$ from state-action-state transitions randomly sampled from the state-action space. Other implementation details are the same as \algo. 

\subsubsection{Choice of Hyper-parameters}\label{sec:choice-param}

In our framework, the hyper-parameters include the Lyapunov convergence rate $\lambda$, the robust parameter $\epsilon$, the weights of the losses $\eta_\mathrm{goal}$, $\eta_\mathrm{pos}$, $\eta_\mathrm{ctrl}$, and the parameters used in training including the learning rate.
$\lambda$ controls the convergence rate of the learned controller. Larger $\lambda$ enables the controller to reach the goal faster, but it also makes the training harder. In our implementation, we choose $\lambda=1-\delta t$, where $\delta t$ is the simulation time step. $\epsilon$ controls the robustness of the learned CLF w.r.t. the Lipschitz constant of the environment, the radius of the trusted tunnel, and the error of the learned dynamics. It should be large enough to satisfy \Cref{thm:clf-converge}, but large $\epsilon$ also makes the training harder. We choose $\epsilon=1.0$ for Inverted Pendulum, Cart Pole, Cart Double Pole, and Neural Lander, and $\epsilon=2.0$ for the F-16 environments. The weights of the losses control the importance of each loss term. Generally, in a simple environment, we tend to use large $\eta_\mathrm{goal}$ and $\eta_\mathrm{pos}$ with small $\eta_\mathrm{ctrl}$, so that the radius of the trusted tunnel can be large and the controller can explore more regions in each iteration, which makes the convergence of our algorithm faster. In a complex environment, however, we tend to use small $\eta_\mathrm{goal}$ and $\eta_\mathrm{pos}$ with large $\eta_\mathrm{ctrl}$. This will limit the divergence between the updated controller and the reference controllers (initial controller and the controller learned in the last iteration) so that the radius of the trusted tunnel won't be so large that the learned CLF is no longer valid. We will further introduce the choice of the weights in Appendix \ref{sec:environments}. The influences of $\lambda$, $\epsilon$, and $\eta_\mathrm{ctrl}$ have been studied in \Cref{sec:ablation}.  The learning rate controls the convergence rate of the training. Large learning rates can make the training faster, but it may also make the training unstable and miss the minimum. We let the learning rate be $3\times 10^{-4}$. 

\subsubsection{Environments}\label{sec:environments}

\paragraph{Inverted Pendulum} 
Inverted pendulum is a standard environment for testing control algorithms. The state of the inverted pendulum is $x=[\theta,\dot\theta]^\top$, where $\theta$ is the angle of the pendulum to the straight-up location, and the control input is the torque. The dynamics is given by $\dot x=f(x)+g(x)u$, with
\begin{equation}\label{eq:loss-clf-ctrl}
    \begin{aligned}
    & f(x)=\left[\begin{array}{c}
        \dot\theta \\
        \frac{g\theta}{L}-\frac{b\dot\theta}{mL^2}
    \end{array}\right]\\
    & g(x)=\left[\frac{1}{mL^2}\right]
    \end{aligned}
\end{equation}
where $g=9.80665$ is the gravitational acceleration, $m=1$ is the mass, $L=1$ is the length, and $b=0.01$ is the damping. We define the goal point at $\xg=[0,0]^\top$. We let the discrete-time dynamics be $x(t+1)=x(t)+\dot x(t)\delta t$, where the simulating time step $\delta t=0.01$. We use reward function $r(x)=2.0-|\theta|$ to train RL algorithms.

We set the initial state with $\theta\in[-0.2,0.2]$ and $\dot\theta\in[-0.2,0.2]$. For the demonstrations, we solve the LQR controller with $Q=I_2$ and $R=1$, where $I_n$ is the $n$-dimensional identity matrix, and add standard deviation $0.1$ and bias $4.0$ to the solution to make it unstable. We collect $20$ trajectories for the demonstrations, where each trajectory has $1000$ time steps. For hyper-parameters in the loss function~\eqref{eq:loss-clf-ctrl}, we use $\eta_\mathrm{goal}=10.0$, $\eta_\mathrm{pos}=10.0$, $\eta_\mathrm{ctrl}=1.0$. 

\paragraph{Cart Pole}
The Cart Pole environment we use is a modification of the InvertedPendulum environment introduced in OpenAI Gym~\citep{gym}. However, the original reward function is not suitable for the stabilization task because there is only one term: ``alive bonus" in the original reward function. Therefore, we change the reward function to be $r=1-\|x\|$ where $x$ is the current state. 

We collect demonstrations using an RL policy that has not fully converged. We collect $20$ trajectories for the demonstrations, where each trajectory has $1000$ time steps. For hyper-parameters in the loss function~\eqref{eq:loss-clf-ctrl}, we use $\eta_\mathrm{goal}=1000.0$, $\eta_\mathrm{pos}=10.0$, $\eta_\mathrm{ctrl}=1000.0$. Note that $\eta_\mathrm{ctrl}$ is large because the control actions in this environment are often tiny.

\paragraph{Cart II Pole}
The Cart II Pole environment we use is a modification of the InvertedDoublePendulum environment introduced in OpenAI Gym~\citep{gym}. To provide more signal for the stabilization task, we change the reward function to be $r=r_\mathrm{origin}-\|q\|$, where $r_\mathrm{origin}$ is the original reward function and $q$ is the current position of the cart. 

We collect demonstrations using an RL policy that has not fully converged. We collect $20$ trajectories for the demonstrations, where each trajectory has $1000$ time steps. For hyper-parameters in the loss function~\eqref{eq:loss-clf-ctrl}, we use $\eta_\mathrm{goal}=100000.0$, $\eta_\mathrm{pos}=5.0$, $\eta_\mathrm{ctrl}=3000.0$. Note that $\eta_\mathrm{ctrl}$ is large because the control actions in this environment are often very small.

\paragraph{Neural Lander}
Neural lander~\citep{shi2019neural} is a widely used benchmark for systems with unknown disturbance. The state of the Neural Lander is $x=[p_x,p_y,p_z,v_x,v_y,v_z]^\top$, with control input $u=[f_x,f_y,f_z]^\top$. $p_x,p_y,p_z$ are the 3D displacements, $v_x,v_y,v_z$ are the 3D velocities, and $f_x,f_y,f_z$ are the 3D forces. The dynamics is given by $\dot x=f(x)+g(x)u$, with 
\begin{equation}
    \begin{aligned}
    & f(x)=\left[v_x,v_y,v_z,\frac{F_{a1}}{m},\frac{F_{a2}}{m},\frac{F_{a3}}{m}-g'\right]^\top\\
    & g(x)=\left[\begin{array}{ccc}
        0 & 0 & 0 \\
        0 & 0 & 0 \\
        0 & 0 & 0 \\
        1/m & 0 & 0 \\
        0 & 1/m & 0 \\
        0 & 0 & 1/m
    \end{array}\right]
    \end{aligned}
\end{equation}
where $g'=9.81$ is the gravitational acceleration, $m=1.47$ is the mass, and $F_a$ is the learned dynamics of the ground effect, represented as a $4$-layer neural network. We define the goal point at $\xg=[0,0,0.5,0,0,0]^\top$.  We let the discrete-time dynamics be $x(t+1)=x(t)+\dot x(t)\delta t$, where the simulating time step $\delta t=0.01$. For the reward function, we use $r(x)=10-\|x\|$.

We set the initial state with $p_x,p_y\in[-2,2]$, $p_z\in[1,2]$, and $v_x,v_y,v_z=0$. For the demonstrations, we use a PD controller
\begin{equation}
    u=\left[\begin{array}{c}
        -8p_x-v_x \\
        -8p_y-v_y \\
        -8p_z-v_z+mg'
    \end{array}\right]
\end{equation}
We collect $20$ trajectories for the demonstrations, where each trajectory has $1000$ time steps. For hyper-parameters in the loss function~\eqref{eq:loss-clf-ctrl}, we use $\eta_\mathrm{goal}=100.0$, $\eta_\mathrm{pos}=50.0$, $\eta_\mathrm{ctrl}=1.0$. 

\paragraph{F-16 Ground Collision Avoidance (GCA)}
F-16~\citep{heidlauf2018verification}\footnote{https://github.com/stanleybak/AeroBenchVVPython (GPL-3.0 license)} is a fixed-wing fighter model. Its state space is 16D including air speed $v$, angle of attack $\alpha$, angle of sideslip $\beta$, roll angle $\phi$, pitch angle $\theta$, yaw angle $\psi$, roll rate $P$, pitch rate $Q$, yaw rate $R$, northward horizontal displacement $p_n$, eastward horizontal displacement $p_e$, altitude $h$, engine thrust dynamics lag $pow$, and three internal integrator states. The control input is 4D including acceleration at z direction, stability roll rate, side acceleration + raw rate, and the throttle command. The dynamics are complex and cannot be described as ODEs, so the authors of the F-16 model provide look-up tables to describe the aerodynamics. The lookup tables describe an approximation of the Lipschitz real dynamics, and also since we simulate the system in a discrete way in the experiments, the look-up table does not violate our assumptions about the real dynamics. We define the goal point at $h=1000$. The simulating time step is $0.02$.

We set the initial state with $v\in[520,560]$, $\alpha=0.037$, $\beta=0$, $\phi=0$, $\theta=-1.4\pi$, $\psi=0.8\pi$, $P\in[-5,5]$, $Q\in[-1,1]$, $R\in[-1,1]$, $p_n=0$, $p_e=0$, $h\in[2600,3000]$, $pow\in[4,5]$. For the demonstrations, we use the controller provided with the model. We collect $40$ trajectories for the demonstrations, where each trajectory has $500$ time steps. For hyper-parameters in the loss function~\eqref{eq:loss-clf-ctrl}, we use $\eta_\mathrm{goal}=100.0$, $\eta_\mathrm{pos}=50.0$, $\eta_\mathrm{ctrl}=50.0$.

\paragraph{F-16 Tracking}
The F-16 Tracking environment uses the same model as the F-16 GCA environment. We define the goal point at $[p_n,p_e,h]=[7500,5000,1500]$, and $\psi=\arctan\frac{p_n}{p_e}$. The simulating time step is $0.02$.

We set the initial state with $v\in[520,560]$, $\alpha=0.037$, $\beta=0$, $\phi\in[-0.1,0.1]$, $\theta\in[-0.1,0.1]$, $\psi\in[-0.1,0.1]$, $P\in[-0.5,0.5]$, $Q\in[-0.5,0.5]$, $R\in[-0.5,0.5]$, $p_n=0$, $p_e=0$, $h=1500$, $pow\in[4,5]$. For the demonstrations, we use the controller provided with the model. We collect $40$ trajectories for the demonstrations, where each trajectory has $500$ time steps. For hyper-parameters in the loss function~\eqref{eq:loss-clf-ctrl}, we use $\eta_\mathrm{goal}=100.0$, $\eta_\mathrm{pos}=50.0$, $\eta_\mathrm{ctrl}=1000.0$. 

\subsection{More Results}\label{sec:more-results}

\begin{figure}
    \centering
    \subfigure[Inverted Pendulum]{\includegraphics[width=.49\textwidth]{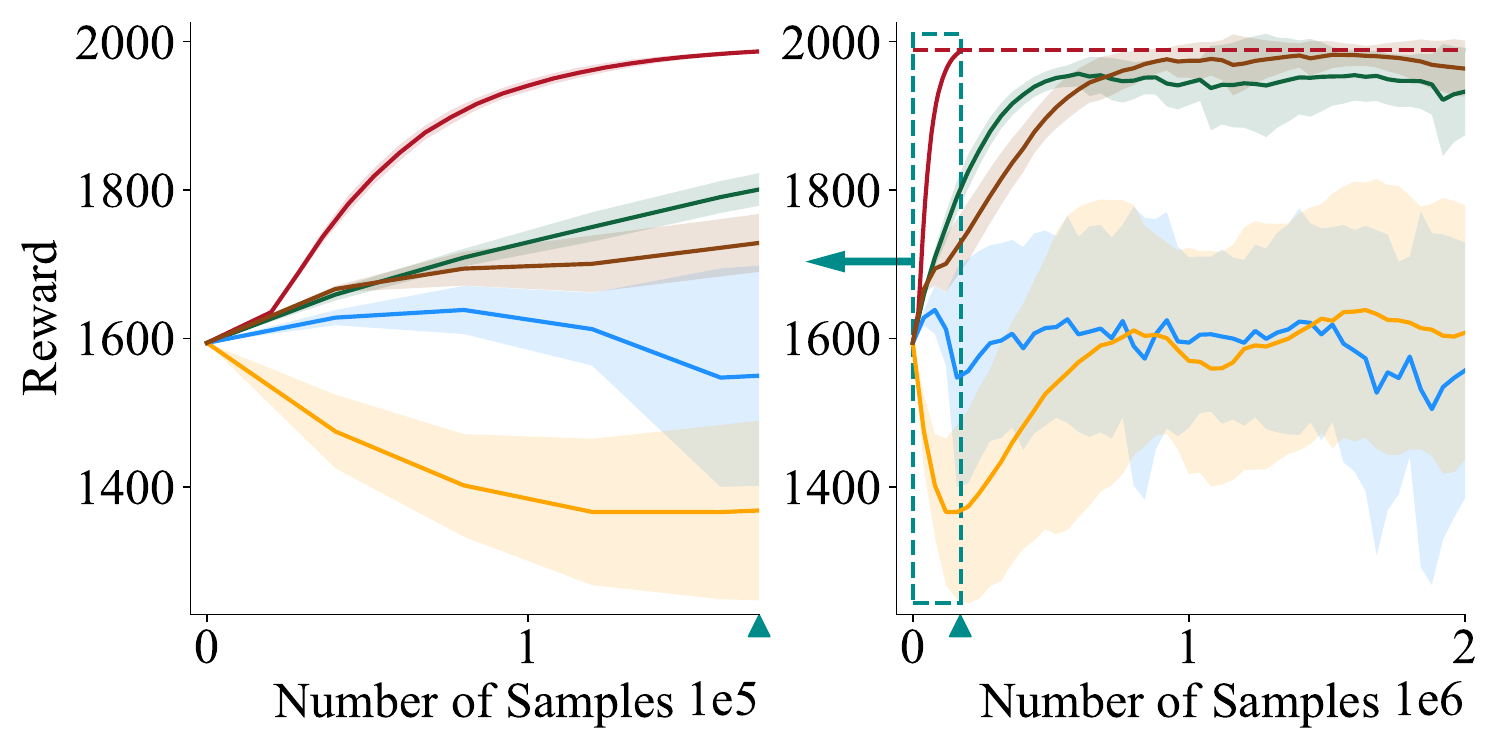}\label{fig:reward-invertedpendulum}}
    \subfigure[Cart Pole]{\includegraphics[width=.49\textwidth]{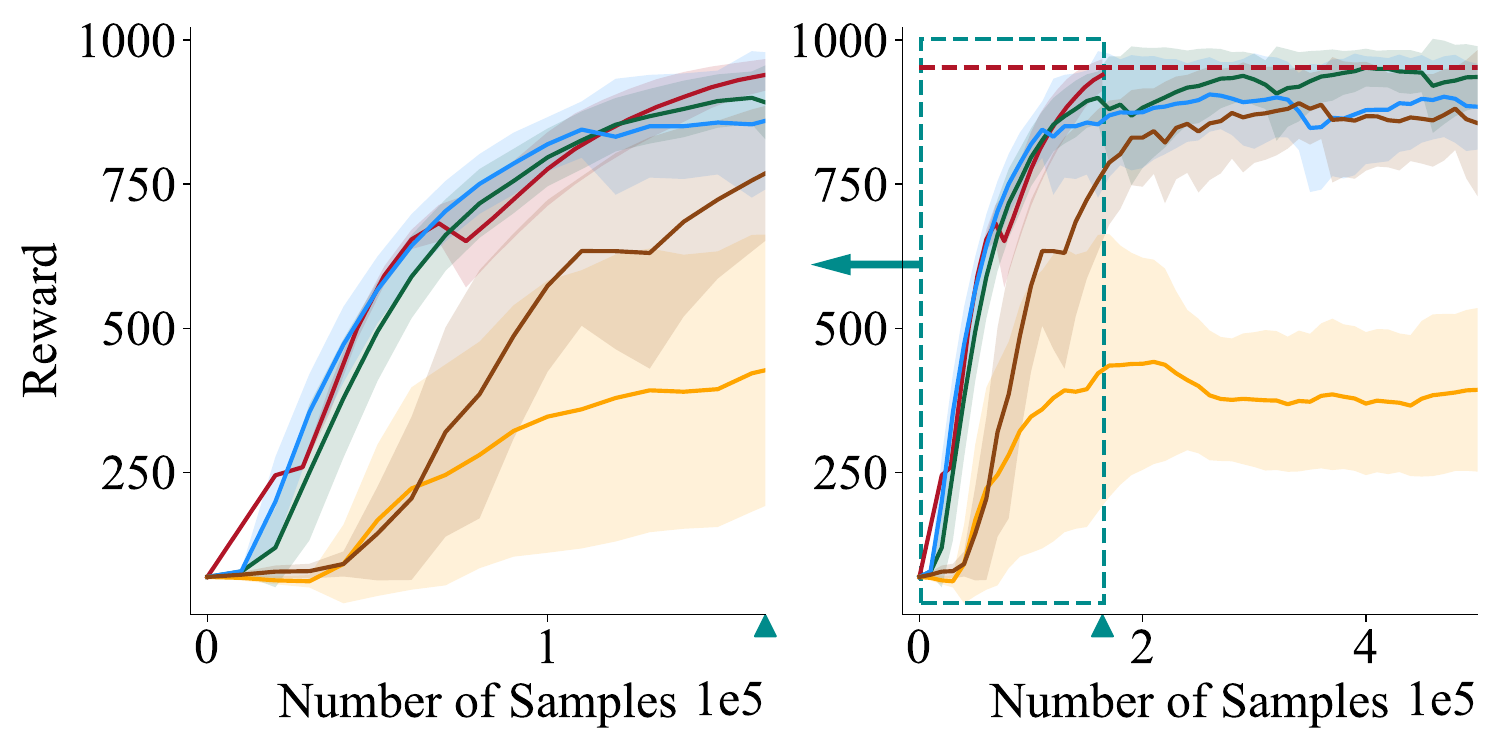}\label{fig:reward-cartpole}}
    \subfigure[Cart II Pole]{\includegraphics[width=.49\textwidth]{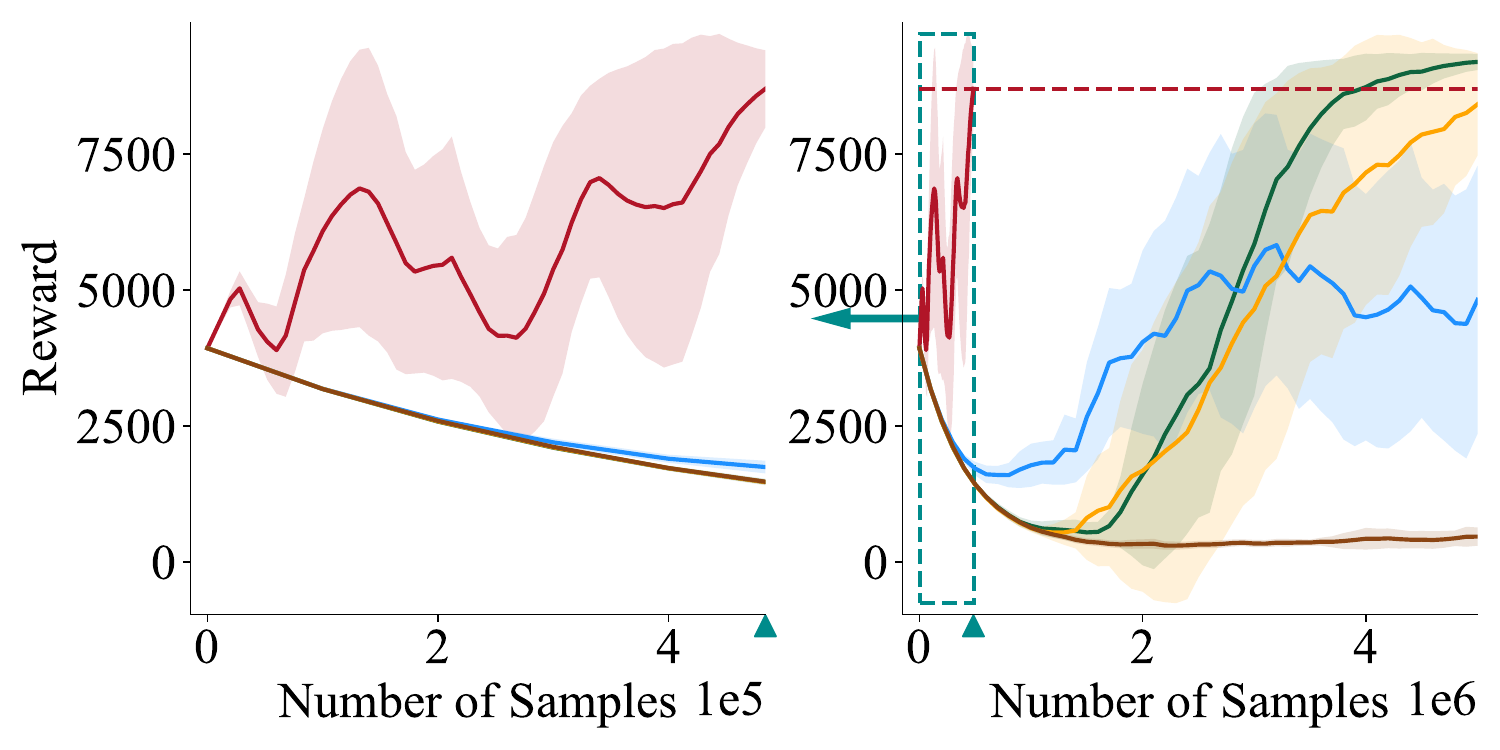}\label{fig:reward-cartdoublepole}}
    \subfigure[Neural Lander]{\includegraphics[width=.49\textwidth]{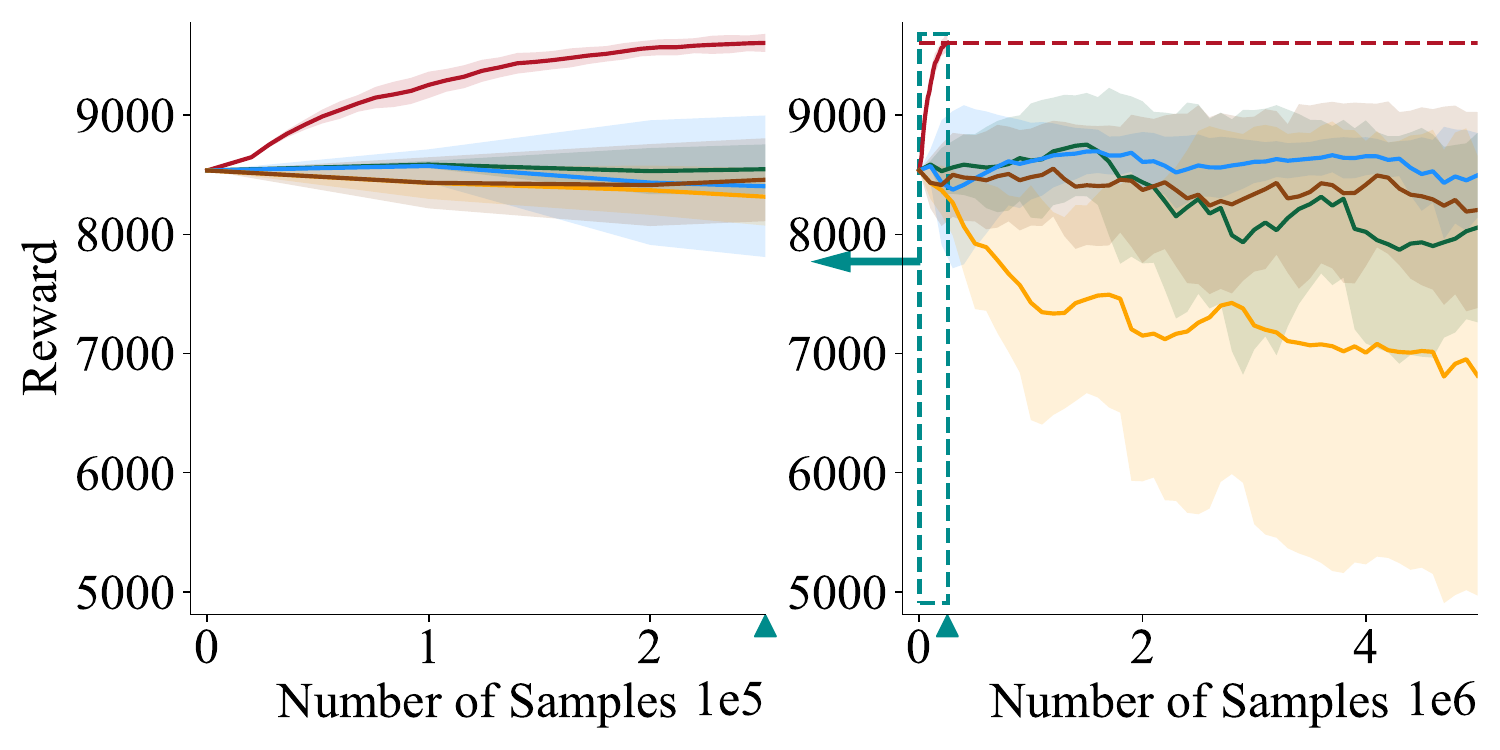}\label{fig:reward-neurallander}}
    \subfigure[F-16 GCA]{\includegraphics[width=.49\textwidth]{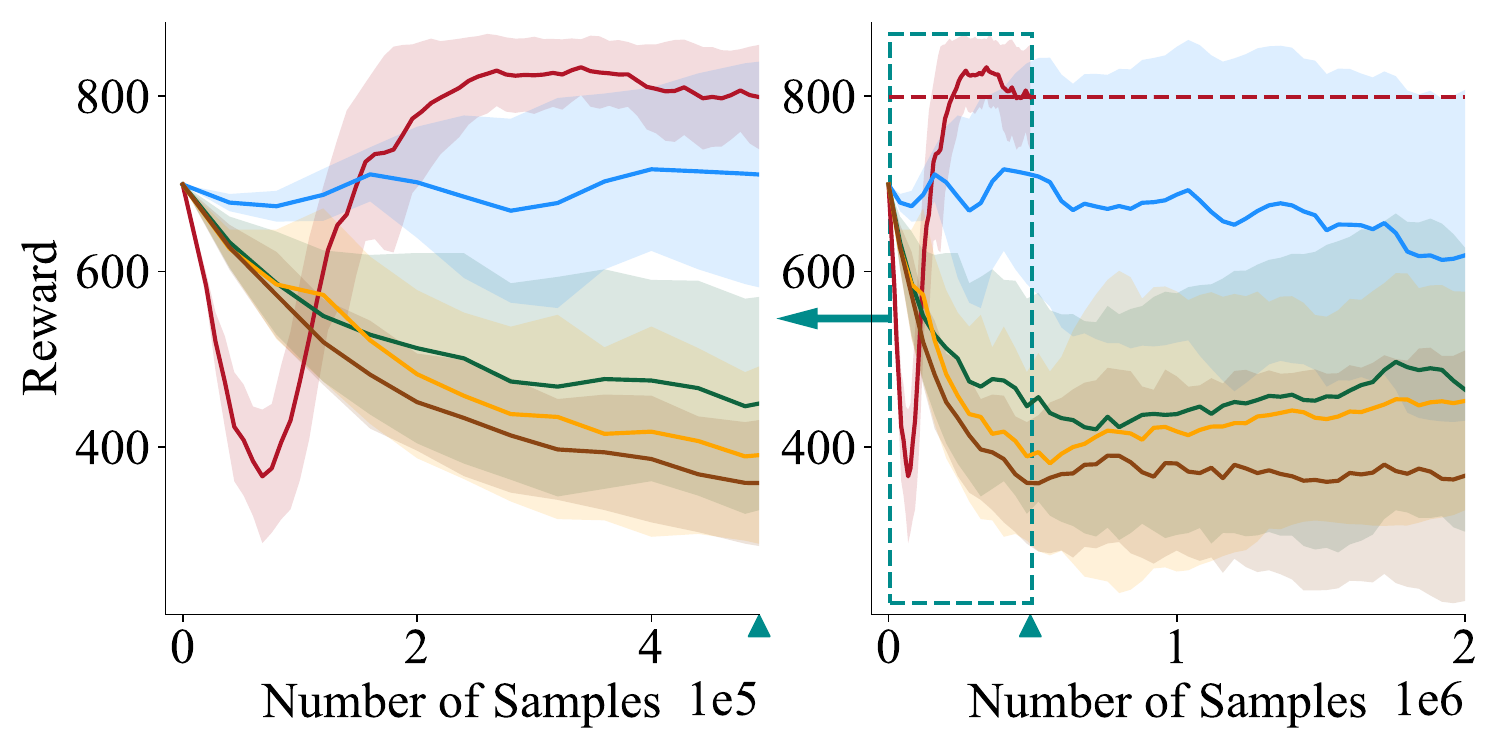}\label{fig:reward-f16gcas}}
    \subfigure[F-16 Tracking]{\includegraphics[width=.49\textwidth]{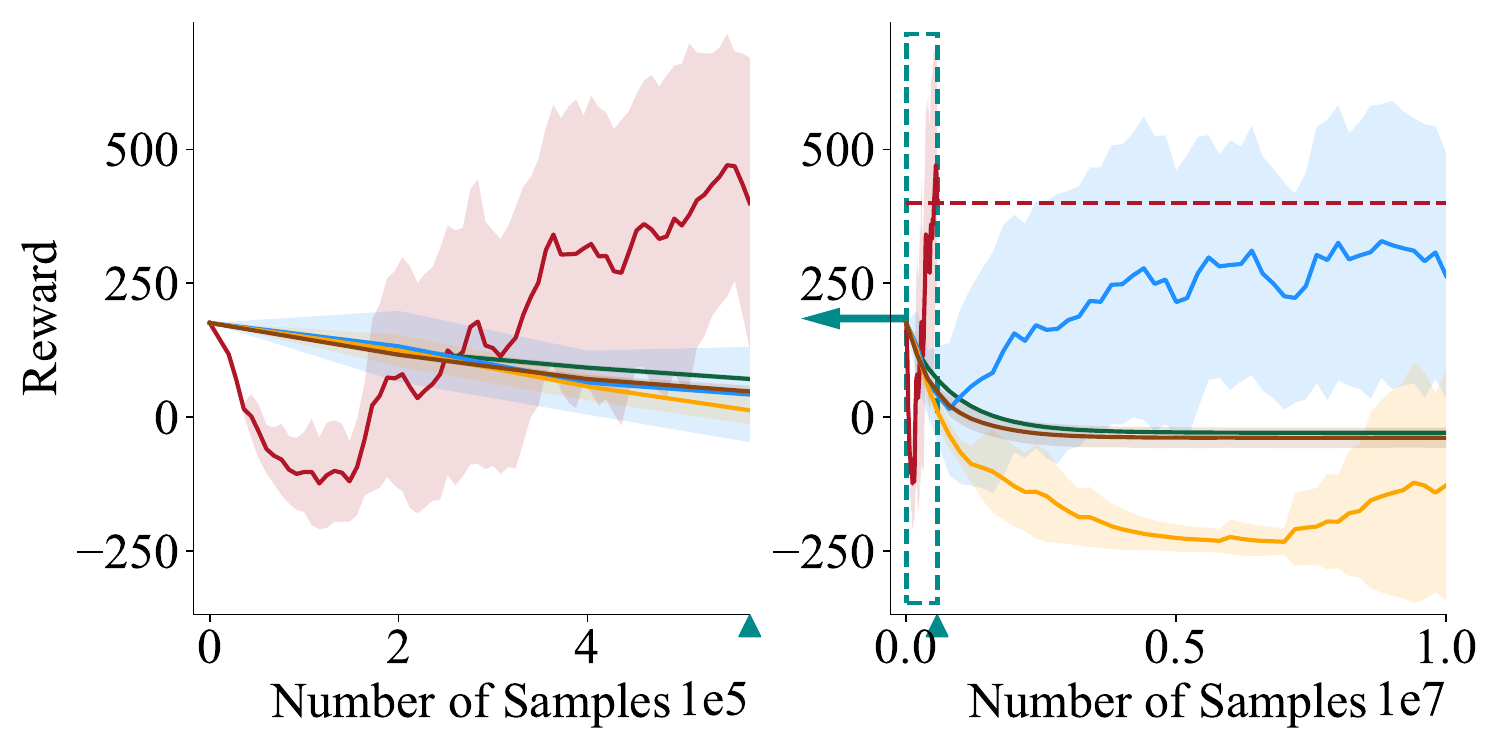}\label{fig:reward-f16tracking}}
    \includegraphics[width=0.98\textwidth]{figs/legend.pdf}
    \caption{The expected return of \algo\ and the baselines with respect to the number of samples. The dashed red line shows the converged reward of \algo. In the right of each subplot, we show the whole curve of the expected return w.r.t. the number of samples, and since our \algo\ converges too fast compared with the baselines, we zoom in the region inside the dashed rectangle and show this region in the left. The triangle on the x-axis shows the number of samples needed by \algo\ to converge. The reward at $0$ number of samples is the reward of demonstrations.}
    \label{fig:reward-curves}
\end{figure}

\paragraph{Training Process}
\Cref{fig:reward-curves} shows the expected rewards and standard deviations of different algorithms with respect to the number of samples used in the training process. We do not show the training processes of CLF-sparse and CLF-dense because they use all the samples from the beginning of the training. It is demonstrated that \algo\ achieves similar rewards in relatively low dimensional systems Inverted Pendulum, Cart Pole, and Cart II Pole, and significantly higher rewards in relatively high dimensional systems Neural Lander, F-16 GCA, and F-16 Tracking. We can also observe that \algo\ converges very fast, using about one order of magnitude fewer samples than the baselines. This supports our argument that the Lyapunov-guided exploration explores only the necessary regions in the state space and thus improving the sample efficiency. Note that the reward of \algo\ decreases a bit before rising up again in the F-16 environments. This is because the reward can only measure the cumulative error w.r.t. the goal. For \algo, before convergence, the goal is not inside the trusted tunnel $\mathcal{H}^\tau$. Therefore, during the exploration stage, the learned controller guides the agent to get out of the trusted tunnel $\mathcal{H}^\tau$. When this happens, there is no guarantee of the agent's behavior, so the cumulative error can increase. However, after convergence, LYGE is able to stabilize the system. The cumulative error is small and the reward is high.

\paragraph{Additional Ablation}
We provide an additional ablation to study the influence of the optimality of the demonstrations in the Neural Lander environment, as a supplement to \Cref{fig:vary-opt} in the main pages. The results are shown in \Cref{fig:vary-opt-lander}, which is consistent with the claims we made about this ablation in the main pages. 

\begin{figure*}[t]
    \centering
    \includegraphics[width=0.2\textwidth]{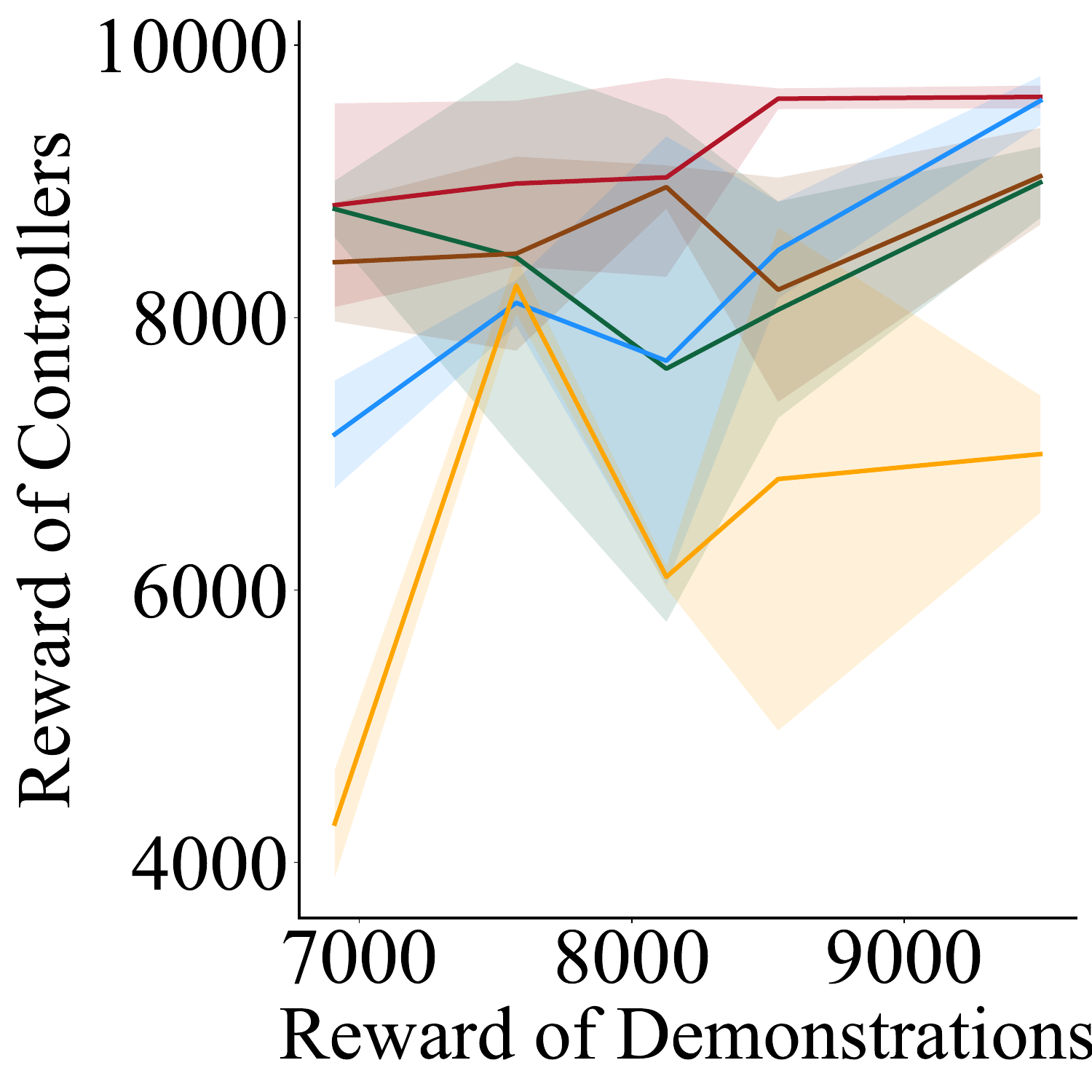}
    \includegraphics[width=.2\textwidth]{figs/legend3.pdf}
    \caption{The reward of the learned controllers w.r.t. the reward of the demonstrations in the Neural Lander environment.}
    \label{fig:vary-opt-lander}
\end{figure*}

\paragraph{Numerical Comparison}
We provide the numerical comparison of the converged rewards of \algo\ and the baselines in \Cref{tab:converge-reward}, corresponding to \Cref{fig:reward-curves}. We can observe that \algo\ performs similarly or outperforms all the baseline methods in all environments. We also provide the numerical results of the ablation studies in \Cref{tab:vary-opt} and \Cref{tab:ablation} corresponding to \Cref{fig:ablations} in the main text. 

\paragraph{Videos for Learned Controllers}
We show the videos of the learned policies of the experiments in the file ‘experiments.mp4’ in the supplementary materials.

\begin{table}[t]
    \caption{Converged rewards and standard deviations of \algo\ and the baselines in the six environments}
    \label{tab:converge-reward}
    \centering
    \small
    \begin{tabular}{c|cccccc}
        \toprule
        Method & Inv Pendulum & Cart Pole & Cart II Pole  & Neural Lander & F-16 GCA & F-16 Tracking \\
        \midrule
        \algo & $\mathbf{1987}\pm2$ & $\mathbf{939}\pm27$ & $8700\pm708$ & $\mathbf{9608}\pm76$ & $\mathbf{798}\pm59$ & $\mathbf{340}\pm272$ \\
        PPO & $1933\pm59$ & $936\pm53$ & $\mathbf{9194}\pm149$ & $8058\pm795$ & $466\pm162$ & $-29\pm1$ \\
        AIRL & $1557\pm173$ & $884\pm74$ & $4816\pm2466$ & $8496\pm355$ & $618\pm188$ & $264\pm229$ \\
        D-REX & $1608\pm172$ & $393\pm142$ & $8414\pm944$ & $6815\pm1847$ & $453\pm124$ & $-127\pm214$ \\
        SSRR & $1964\pm38$ & $856\pm127$ & $463\pm168$ & $8205\pm823$ & $368\pm143$ & $-39\pm19$ \\
        Demo & $1594\pm58$ & $69\pm64$ & $3929\pm1056$ & $8537\pm96$ & $699\pm40$ & $176\pm478$ \\
        \bottomrule
    \end{tabular}
\end{table}

\begin{table}[t]
    \centering
    \caption{Converged reward and number of samples used of \algo\ and the converged reward of baselines in the Inverted Pendulum environment with demonstrations with different optimality.}
    \begin{tabular}{c|ccccccc}
        \toprule
        Demonstrations & $1837$ & $1773$ & $1688$ & $1594$ \\
        \midrule
        \algo & $\mathbf{1995.63}\pm2.69$ & $\mathbf{1995.38}\pm3.71$ & $\mathbf{1995.00}\pm3.16$ & $\mathbf{1995.52}\pm3.05$ \\
        PPO & $1958.74\pm2.15$ & $1949.79\pm8.93$ & $1978.03\pm1.52$ & $1943.27\pm5.75$ \\
        AIRL & $1831.78\pm47.37$ & $1875.50\pm23.55$ & $1616.20\pm4.62$ & $1604.92\pm5.97$ \\
        D-REX & $1888.75\pm3.61$ & $1772.05\pm8.98$ & $1831.68\pm6.50$ & $1646.63\pm6.56$ \\
        SSRR & $1981.96\pm0.65$ & $1970.13\pm0.84$ & $1952.56\pm1.34$ & $1982.67\pm0.61$ \\
        Converged Samples  & $45333\pm16653$ & $58667\pm12220$ & $58667\pm4619$ & $77333\pm12220$\\
        \midrule
        Demonstrations & $1498$ & $1400$ & $1302$ \\
        \midrule
        \algo & $\mathbf{1995.28}\pm2.81$ & $\mathbf{1995.41}\pm3.12$ & $\mathbf{1992.91}\pm3.22$ \\
        PPO & $1962.92\pm2.80$ & $1952.21\pm2.34$ & $1916.23\pm2.95$ \\
        AIRL & $1399.69\pm11.65$ & $1578.24\pm11.39$ & $1154.61\pm7.45$ \\
        D-REX & $1620.61\pm12.01$ & $1649.49\pm6.22$ & $1053.79\pm3.22$ \\
        SSRR & $1952.05\pm1.73$ & $1756.77\pm62.98$ & $1944.05\pm5.48$ \\
        Converged Samples & $101333\pm16653$ & $104000\pm21166$ & $98667\pm4619$ \\
        \bottomrule
    \end{tabular}
    \label{tab:vary-opt}
\end{table}

\begin{table}[t]
    \centering
    \begin{tabular}{c|ccccc}
        \toprule
        $\epsilon$ & $0.01$ & $0.10$ & $1.00$ & $10.00$ & $100.00$ \\
        Reward & $895\pm1$ & $988\pm6$ & $969\pm22$ & $96\pm36$ & $17\pm1$ \\
        \midrule
        $\lambda$ & $0.01$ & $0.10$ & $1.00$ & $10.00$ & $100.00$ \\
        Reward & $147\pm12$ & $465\pm430$ & $969\pm22$ & $896\pm1$ & $895\pm1$ \\
        \midrule
        $\eta_\mathrm{ctrl}$ & $10$ & $100$ & $1000$ & $10000$ & $100000$ \\
        Reward & $342\pm475$ & $619\pm500$ & $969\pm22$ & $984\pm4$ & $894\pm1$\\
        \bottomrule
    \end{tabular}
    \caption{Numerical results of the ablations on $\epsilon$, $\lambda$, and $\eta_\mathrm{ctrl}$. }
    \label{tab:ablation}
\end{table}

\section{Discussions}\label{sec:app-discussions}

\subsection{Verification of the Learned CLFs}
The focus of this paper is to use the learned CLF to guide the exploration and synthesize feedback controllers for high-dimensional unknown systems. The experimental results show that the learned controllers are goal-reaching, and we find that the learned CLF satisfies the conditions in the majority of the trusted tunnel, and the learned CLF can successfully guide the controller to explore the useful subset of the state space. However, we do not claim that our learned CLF is formally verified. 
Although the theories we provide suggest that the learned CLF satisfies the CLF conditions \eqref{eq:clf-cond} under certain assumptions, these assumptions may not be satisfied in experiments. For example, we cannot theoretically guarantee that the loss always $\epsilon'$-robustly converges on any system. 
If we want to verify the learned CLF is valid in the whole trusted tunnel, additional verification tools are needed, including SMT solvers~\citep{gao2012delta,chang2019neural}, Lipschitz-informed sampling methods~\citep{bobiti2018automated}, etc. However, these verification tools are known to have poor scalability and thus generally used in systems with dimensions less than $6$~\citep{chang2019neural,zhou2022neural}. Scalable verification for the learned CLFs remains an open problem.



\subsection{Possible Future Directions}

Our algorithm can also benefit from the verification tools. For instance, 
SMT solvers allow us to find counterexamples to augment the training data to make our learned CLF converge faster. In addition, \emph{almost Lyapunov functions}~\citep{liu2020almost} show that the system can still be stable even if the Lyapunov conditions do not hold everywhere. We are excited to explore these possible improvements in our future work.

\section{Details about the Extensions}\label{sec:extension-detail}

\subsection{Learning Control Contraction Matrices with Unknown Dynamics}

In Section 7 in the main text, we introduced that our algorithm can also be directly applied to learn Control Contraction Matrices (CCMs) in environments with unknown dynamics. Here we provide more details. 

We consider the control-affine systems 
\begin{equation}
    \dot x=f(x)+g(x)u
\end{equation}
where $x\in\mathcal{X}\subseteq\mathbb{R}^{n_x}$ is the state, $u\in\mathcal{U}\subseteq\mathbb{R}^{n_u}$ is the control input. The tracking problem we consider is to design a controller $u=\pi(x,x^*,u^*)$, such that the controlled trajectory $x(t)$ can track \emph{any} target trajectory $x^*(t)$ generated by some reference control $u^*(t)$ when $x(0)$ is near $x^*(0)$. 

CCMs are widely used to provide contraction guarantees for tracking controllers. A fundamental theorem in CCM theory~\citep{manchester2017control} says that if there exists a metric $M(x)$ and a constant $\lambda>0$, such that
\begin{subequations}
    \begin{align}
    & g_\bot^\top\left(-\partial_fW(x)+\reallywidehat{\frac{\partial f(x)}{\partial x}W(x)}+2\lambda W(x)\right)g_\bot\prec0\label{eq:loss-c1}\\
    & g_\bot^\top\left(\partial_{g_j}W(x)-\reallywidehat{\frac{\partial g_j(x)}{\partial x}W(x)}\right)g_\bot=0,\quad j=1,...,n_u\label{eq:loss-c2}
    \end{align}
\end{subequations}
where $g_\bot(x)$ is an annihilator matrix of $g(x)$ satisfying $g_\bot^\top(x)g(x)=0$, $W(x)=M(x)^{-1}$ is the dual metric, $g_j$ is the $j$-th column of matrix $g$, and for a matrix $P$, $\widehat{P}=P+P^\top$, then there exists a controller $u=\pi(x,x^*,u^*)$, such that the controlled trajectory $x(t)$ will converge to the reference trajectory $x^*(t)$ exponentially. Such controller can be find by satisfying the following condition:
\begin{equation}\label{eq:loss-cu}
    \dot M+\reallywidehat{M(A+gK)}+2\lambda M\prec 0
\end{equation}
where $A=\frac{\partial f}{\partial x}+\sum_{i=1}^{n_u} u^i\frac{\partial g_i}{\partial x}$, $u^i$ is the $i$-th element of the vector $u$, and $K=\frac{\partial u}{\partial x}$. 

We use a similar framework as \algo\ to learn the CCM in unknown systems. Since CCM theories require the environment dynamics to be control-affine, we change the model of the dynamics to be $\hat h^\tau_{\alpha,\beta}(x,u)=f^\tau_\alpha(x)+g^\tau_\beta(x)u$, where $f^\tau_\alpha(x): \mathbb{R}^{n_x}\rightarrow\mathbb{R}^{n_x}$ and $g^\tau_\beta(x): \mathbb{R}^{n_x}\rightarrow\mathbb{R}^{n_x}\times\mathbb{R}^{n_u}$ are neural networks with parameters $\alpha$ and $\beta$. Given imperfect demonstrations, we still first use imitation learning to learn an initial controller $\pi_\mathrm{init}(x,x^*,u^*)$, and fit the local model $\hat h^0_{\alpha,\beta}(x,u)$. In order to find $g_\bot(x)$, we need the learned $g^\tau_\beta(x)$ to be sparse so that we can hand-craft $g_\bot(x)$ for it, so we add the Lasso regression term in the loss $\mathcal{L}_\mathrm{dyn}$:
\begin{equation}
    \begin{aligned}
    \mathcal{L}^\tau_\mathrm{dyn}(\alpha,\beta)&=\frac{1}{N}\sum_{x(t),u(t),x(t+1)\in\mathcal{D}^\tau_x} \left\|x(t+1)-x(t)-\hat h^\tau_{\alpha,\beta}(x(t),u(t))\right\|^2\\
    &+\mu_\mathrm{dyn}(\|\alpha\|^2+\|\beta\|^2+\|\beta\|_1)
    \end{aligned}
\end{equation}
where $\|\beta\|_1$ is the $1$-norm of $\beta$. Then, we jointly learn the controller and the corresponding CCM inside $\mathcal{H}^\tau$. We parameterize the controller and the dual metric using neural networks $\pi^\tau_\phi(x,x^*,u^*)$ and $W^\tau_\theta(x)$ with parameters $\phi$ and $\theta$, and train them by replacing the CLF loss $\mathcal{L}^\tau_\mathrm{CLF}$ in the main text with the following loss:
\begin{equation}
    \mathcal{L}^\tau_\mathrm{CCM}=\frac{1}{N}\sum_{(x,x^*,u^*)\in\mathcal{D}^\tau_x}\left[L_\mathrm{PD}(-C_1(x;\theta))+\sum_{j=1}^{n_u}\|C_2^j(x;\theta)\|_F+L_\mathrm{PD}(-C_u(x,x^*,u^*;\phi))\right]
\end{equation}
where $C_1(x;\theta)$, $C_2(x;\theta)$, and $C_u(x,x^*,u^*;\phi)$ are the LHS of \Cref{eq:loss-c1}, \Cref{eq:loss-c2}, and \Cref{eq:loss-cu}, respectively.
$\|\cdot\|_F$ is the Frobenius norm. $L_\mathrm{PD}$ is the loss function to make its input positive definite. In our implementation, we use $L_\mathrm{PD}(\cdot)=\frac{1}{N}\sum\mathrm{ReLU}(\lambda(\cdot))$, where $\lambda(\cdot)$ is the eigenvalues. Note that $\theta$ is not a parameter of $C_u$ since we only use $C_u$ to find the controller. For more detailed discussions of the loss functions, one can refer to \cite{sun2021learning} and \cite{chou2021model}. Once we have the learned controller, we apply it in the environment to collect more data and enlarge the trusted tunnel $\mathcal{H}$, following the same process of \algo. We repeat this process several times until convergence. 

\subsection{Experimental Details of CCM}

\subsubsection{Implementation Details}

We implement our algorithm using the PyTorch framework~\citep{paszke2019pytorch} based on the CCM repository~\footnote{https://github.com/MIT-REALM/ccm}~\citep{sun2021learning}. The neural networks in the dynamics model $f^\tau_\alpha(x)$ and $g^\tau_\beta(x)$ have two hidden layers with size $128$ and $\mathrm{Tanh}$ as the activation function. We parameterize our controller using
\begin{equation}
    \pi^\tau_\phi(x,x^*,u^*)=\omega^\tau_2(x,x^*)\cdot\tanh{(\omega^\tau_1(x,x^*)\cdot(x-x^*))}+u^*
\end{equation}
where $\omega^\tau_1(x,x^*)$ and $\omega^\tau_2(x,x^*)$ are two neural networks with two hidden layers with size $128$ and $\mathrm{Tanh}$ as the activation function. Therefore, we have $x=x^*\implies u=u^*$ by construction. We model the dual metric using
\begin{equation}
    W^\tau_\theta(x)=C^\tau(x)^\top C^\tau(x)+\underline{\omega}I
\end{equation}
where $C(x)\in\mathbb{R}^{n_x\times n_x}$ is a neural networks with two hidden layers with size $128$ and $\mathrm{Tanh}$ as the activation function, $I$ is the identity matrix and $\underline{\omega}$ is the minimum eigenvalue. By construction, $W(x)$ is symmetric. We use ADAM as the optimizer with learning rate $3\times 10^{-4}$ to optimize the parameters. For the hyper-parameters, we set the convergence rate of CCM $\lambda=0.5$, and the minimum eigenvalue $\underline{\omega}=0.1$.

\begin{figure}[t]
    \centering
    \includegraphics[width=.49\textwidth]{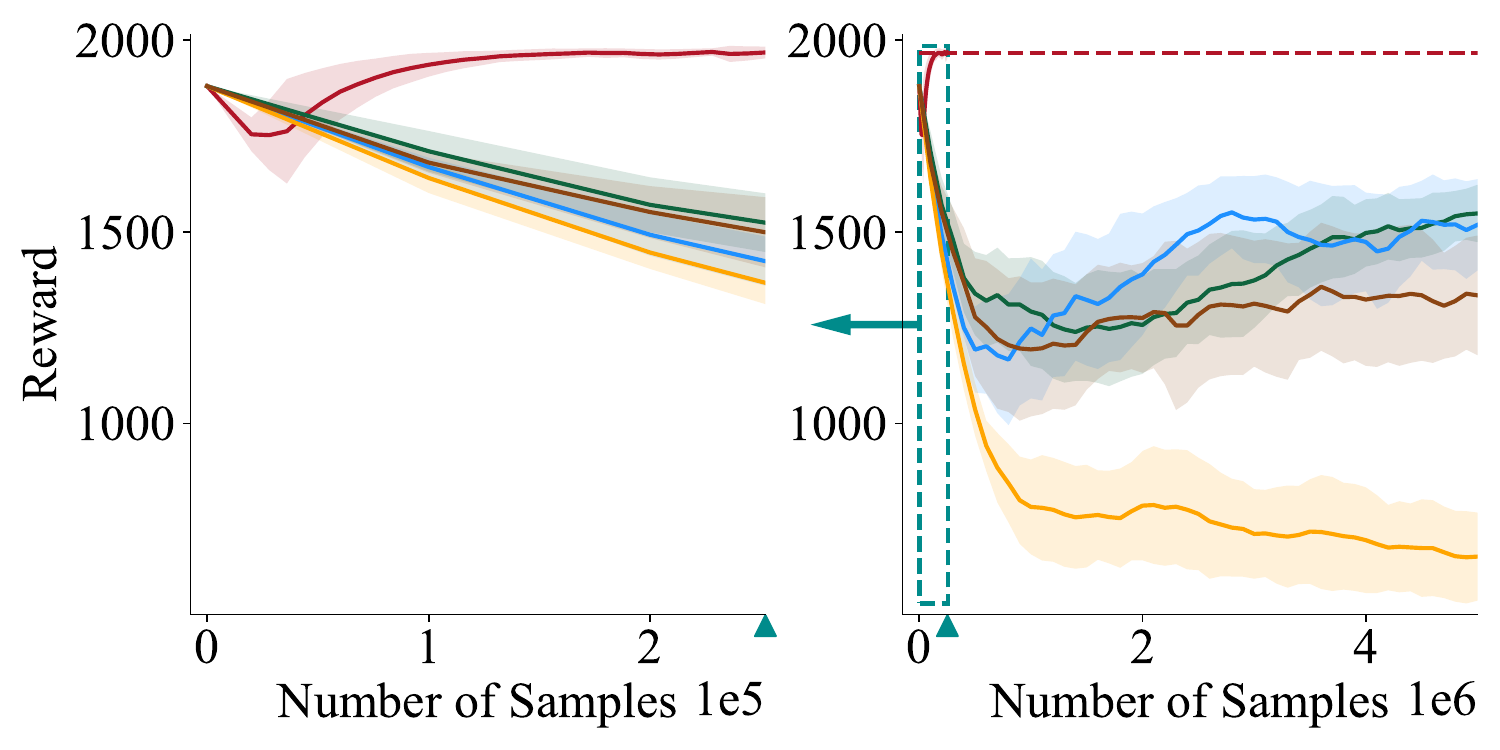}
    \includegraphics[width=.15\textwidth]{figs/legend3.pdf}
    \caption{The expected return w.r.t. the number of samples in Dubins car path tracking environment.}
    \label{fig:dubins-car-reward}
\end{figure}

\begin{table}[t]
    \caption{Converged Reward, Number of Samples, and Tracking Error of our algorithm and baselines in the Dubins car path tracking environment}
    \label{tab:ccm}
    \centering
    \begin{tabular}{c|cccc}
        \toprule
        Method & Converged Reward & Number of Samples (k) & Mean Tracking Error \\
        \midrule
        Ours & $\mathbf{1966.91}\pm15.42$ & $\mathbf{252}$ & $\textbf{0.0228}\pm0.0073$ \\
        PPO & $1547.52\pm74.74$ & $5000$ & $0.597\pm0.157$ \\
        AIRL & $1517.81\pm119.11$ & $5000$ & $0.494\pm0.255$ \\
        D-REX & $653.46\pm114.71$ & $5000$ & $0.701\pm0.166$ \\
        SSRR & $1334.08\pm155.85$ & $5000$ & $0.476\pm0.214$ \\
        \bottomrule
    \end{tabular}
\end{table}

\subsubsection{Environment}

We test our algorithm in a Dubins car path tracking environment. The state of the Dubins car is $x=[p_x,p_y,\psi,v]^\top$, where $p_x$, $p_y$ are the position of the car, $\psi$ is the heading, and $v$ is the velocity. The control input is $u=[\omega,a]\top$ where $\omega$ is the angular acceleration and $a$ is the longitudinal acceleration. The dynamics is given by $\dot x=f(x)+g(x)u$, with
\begin{equation}
    \begin{aligned}
     &f(x)=[v\cos\psi,v\sin\psi,0,0]^\top\\
     &g(x)=\left[\begin{array}{cc}
        0 & 0 \\
        0 & 0 \\
        1 & 0 \\
        0 & 1 \\
     \end{array}\right]
    \end{aligned}
\end{equation}

We set the initial state with $p_x,p_y\in[-0.2,0.2]$, $\psi\in[-0.5,0.5]$, and $v\in[0,0.2]$. For the demonstrations, we use the LQR controller solved with the error dynamics. We collect $20$ trajectories for the demonstrations with randomly generated reference paths, where each trajectory has $1000$ time steps. For the reward function, we use $r(x)=2-\|(p_x,p_y)-(p_x^*,p_y^*)\|$, where $(p_x^*,p_y^*)$ is the position on the reference path.

\subsubsection{More Results}

\paragraph{Training Process}
In \Cref{sec:extensions} we compare the tracking error of our algorithm and the baselines. Here in \Cref{fig:dubins-car-reward} we also provide the expected rewards and standard deviations of our algorithm and the baselines w.r.t. the number of samples used in the training process. We can observe that our algorithm reaches the highest reward and converges much faster than the baselines. 

\paragraph{Numerical Results}
We provide more detailed numerical results here. In \Cref{tab:ccm}, we show the converged reward, the number of samples used in training, and the mean tracking error of our algorithm and the baselines. We can observe that our algorithm achieves the highest reward and the lowest mean tracking error, with a large gap compared with other algorithms. In addition, the samples we used in training are more than one order of magnitude less than other algorithms.

\end{document}